%% file: ms.tex
\documentclass[11pt,final]{article}

\usepackage{fullpage}
\usepackage{amsmath,amssymb,amsthm}
\usepackage{lmodern}
\usepackage[numbers]{natbib}
\usepackage[colorlinks=true]{hyperref}
\hypersetup{
     colorlinks  = true,
     urlcolor    = teal,
	 citecolor   = teal,
	 linkcolor   = red
}

\usepackage[ruled,noend,linesnumbered]{algorithm2e} % For algorithms
\renewcommand{\algorithmcfname}{ALGORITHM}
\SetAlFnt{\small}
\SetAlCapFnt{\small}
\SetAlCapNameFnt{\small}
\SetAlCapHSkip{0pt}
\IncMargin{-\parindent}

\usepackage[utf8]{inputenc} % allow utf-8 input
\usepackage[T1]{fontenc}    % use 8-bit T1 fonts
\usepackage{lmodern}
\usepackage[numbers]{natbib}
\usepackage{url}            % simple URL typesetting
\usepackage{amsfonts}       % blackboard math symbols
\usepackage{nicefrac}       % compact symbols for 1/2, etc.
\usepackage{microtype}      % microtypography
\usepackage{booktabs}       % professional-quality tables

\input{defs}
\renewcommand{\citet}[1]{\citeauthor{#1} \cite{#1}}
\begin{document} 
\title{Regression Equilibrium}

\author{Omer Ben{-}Porat%
\thanks{%
    {Technion - Israeli Institute of Technology (\url{omerbp@campus.technion.ac.il})}}
\and Moshe Tennenholtz%
\thanks{%
    {Technion - Israeli Institute of Technology (\url{moshet@ie.technion.ac.il})}}
}
\maketitle

\begin{abstract}
Prediction is a well-studied machine learning task, and prediction algorithms are core ingredients in online products and services.
Despite their centrality in the competition between online companies who offer prediction-based products, the \textit{strategic} use of prediction algorithms remains unexplored. The goal of this paper is to examine strategic use of prediction algorithms.
We introduce a novel game-theoretic setting that is based on the PAC learning framework, where each player (aka a prediction algorithm aimed at competition) seeks to maximize the sum of points for which it produces an accurate prediction and the others do not. We show that algorithms aiming at generalization may wittingly mispredict some points to perform better than others on expectation. We analyze the empirical game, i.e., the game induced on a given sample, prove that it always possesses a pure Nash equilibrium, and show that every better-response learning process converges. Moreover, our learning-theoretic analysis suggests that players can, with high probability, learn an approximate pure Nash equilibrium for the whole population using a small number of samples.  
\end{abstract}

\section{Introduction}
Prediction plays an important role in twenty-first century economics. In a prediction task, an algorithm is given a sequence of examples composed of labeled instances, and its goal is to learn a general rule that maps instances to labels. With the recent data explosion, commercial companies can, like never before, collect massive amounts of data and employ sophisticated machine learning algorithms to discover patterns and seek connections between different observations. For instance, after examining a sufficient number of apartments -- their characteristics and selling prices -- real estate experts may attempt to accurately predict the selling price of a new, unseen apartment. 
Typically, the quality of a prediction algorithm is measured by its success in predicting the value of an unlabeled (or unseen) instance.

However, ubiquitously prediction is not carried out in isolation. 
For revenue-seeking companies, prediction is another tool that can be exploited to increase revenue. To illustrate, consider several competing real estate experts, who provide prediction service for the selling value of apartments on their websites. These experts gain directly from user traffic to their websites, and hence aim to attract as many users as possible. A user, after receiving the experts' predictions and selling his\footnote{For ease of exposition, third-person singular pronouns are ``he'' for a user and ``she'' for a player.} apartment, can evaluate which experts were accurate\footnote{Crucially, we assume that the predicted values and the actual ones are independent. This is arguably the case if experts are external, and not buyers nor real-estate brokers. Consequently, users care for accurate prediction and not over-estimation.} and which were not; he will typically decide that an expert is accurate based on his experience and/or the experience of his family and friends. That user may decide to have future interaction with one/some of the accurate experts, or recommend his friends to interact with them; thus, providing an accurate prediction to a user translates to higher revenue. Interestingly, maximizing revenue and minimizing discrepancy do not coincide in this example, as further illustrated in Figure \ref{fig:motivation}, suggesting that prediction algorithms in competition should optimize revenue explicitly, and not other measures that affect revenue only indirectly. Despite the intuitive clarity of this tradeoff and the enormous amount of work done on prediction in the machine learning and statistical learning communities, far too little attention has been paid to the study of prediction in the context of \textit{competition}. 

\begin{figure}[t]
\centering
\includegraphics[scale=.7]{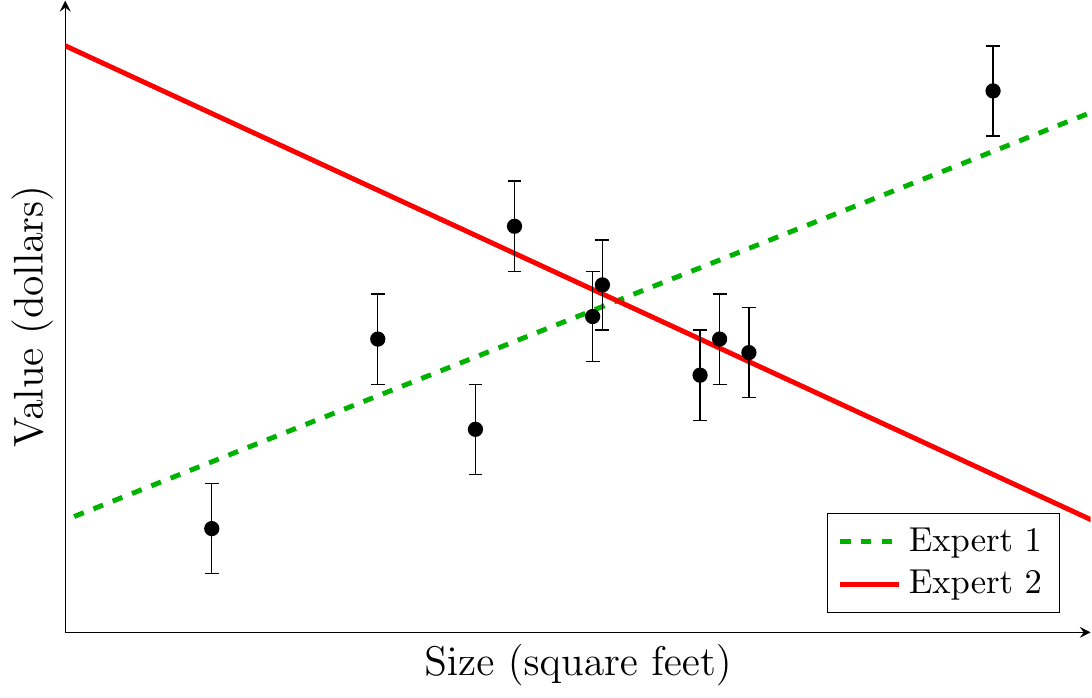}
\caption{A case where minimizing the square error is non-optimal. Each point is an instance-value pair, and  
the user associated with each point considers a prediction as accurate if it lies inside the segment engulfing that point. The green (dashed) line represents the predictions of expert 1, employing the linear least squares estimators and being accurate on 20\% of the points. The red (solid) line represents the predictions of expert 2, providing an accurate prediction on 60\% of the points, and completely ignores the increasing trend of the data. Expert 1 could be considered accurate by more points if she deviates to the regression line of expert 2 (among other deviations).
\label{fig:motivation}}
\end{figure}

In this paper we study how the selection of prediction algorithms is affected by strategic behavior in a competitive setting, using a game-theoretic lens. We consider a space of users, where each user is modeled as a triplet $(x,y,t)$ of an instance, a label and a threshold, respectively. An instance encodes public information describing the user, e.g., a real vector with his apartment's characteristics; the label is the value that should be predicted, e.g., the actual selling value of his apartment (that is only revealed later); and the threshold is the ``distance'' he is willing to accept between a proposed prediction and his label. Namely, a prediction $f(x)$ is said to be \textit{accurate} w.r.t. the user associated with $(x,y,t)$ if $\abs{f(x) - y}$ is less than or equal to $t$. 
In such a case, the user is \textit{satisfied} and willing to adopt the prediction, granting one monetary unit to the expert who produced the accurate prediction. If a user is satisfied with several predictions (of several experts), he selects one uniformly at random. Indeed, the user-model we adopt is aligned with the celebrated ``Satisficing" principle of \citet{simon1956rational}, and other widely-accepted models in the literature on choice prediction, e.g., the model of selection based on small samples \cite{barron2003small,erev2010choice}.
Several players (companies or experts) are equipped with infinite strategy spaces, or hypothesis classes in learning-theoretic terminology. A player's strategy space models the possible predictive functions she can employ. Players are competing for the users, and a player's payoff is the expected number of users who find her predictions to be accurate, in the aforementioned sense. To model uncertainty w.r.t. the labels and thresholds, we use the PAC-learning framework of \citet{valiant1984theory}. We assume user distribution is unknown, but the players have access to a sequence of examples, containing instances, labels and thresholds, with which they should optimize their payoffs w.r.t. the unknown underlying user distribution.

From a machine learning perspective, we now face the challenge of what would be a good prediction algorithm profile, i.e., a set of algorithms for the players such that no player would deviate from her algorithm assuming the others all stick to their algorithms. Indeed, such a profile of algorithms determines a \textit{pure Nash equilibrium} (PNE), a powerful solution concept which rarely exists in games. An important question in this regard is whether such a profile exists. An accompanying question is whether learning dynamics in which players may change their prediction algorithms to better respond to others would converge. 
Therefore, we ask:

\begin{itemize}
\item Does a PNE exist?
\item Will the players be able to find it efficiently with high probability using better-response dynamics?
\end{itemize}

We prove that the answer to both questions is in the affirmative. We first show that when the capacity of each strategy space is bounded (i.e., finite pseudo-dimension), players can learn payoffs from samples. Namely, we show that the payoff function of each player uniformly converges over all possible strategy profiles (that include strategies of the other players); thus, with high probability a player's payoff under any strategy profile is not too distant from her empirical payoff. 
Later, we show that an empirical PNE always exists, i.e., a PNE of the game induced by the empirical sample distribution. Moreover, we show that any learning dynamics in which players improve their payoff by more than a non-negligible quantity converges fast to an approximate PNE. 
Using the two latter results, we show an interesting property of the setting: the elementary idea of sampling and better-responding according to the empirical distribution until convergence leads to an approximate PNE of the game on the whole population. We analyze this learning process, and formalize the above intuition via an algorithm that runs in polynomial time in the instance parameters, and returns an approximate PNE with high probability.

To further exemplify the setting, we then consider the case of players employing linear regression. We modify the algorithm of \citet{porat2017best}, proposed for a related setting, to devise an algorithm for the best linear response of a player. We conduct simulations on various synthetic datasets to visualize the structure of PNEs in several settings, and to perform sensitivity analysis. Finally, we discuss two extensions. We consider the case of infinite capacities, and demonstrate that non-learnability can occur even if the user distribution is known to all players. We also study a natural variant of the model, where each user grants one monetary unit to the most accurate player, and show that learning and PNE existence are no longer guaranteed.

\subsection{Related work}
The intersection of research in game theory and machine learning has increased rapidly in recent years. Sample empowered mechanism design is a fruitful line of research. For example, \citet{cole2014sample}, \citet{GonczarowskiN17}, \citet{morgenstern2015pseudo} reconsider  auctions where the auctioneer can sample from bidder valuation functions, thereby relaxing the assumption of prior knowledge on bidder valuation distribution \cite{myerson1981optimal}.
Empirical distributions also play a key role in other lines of research  \cite{althofer1994sparse,babichenko2016empirical,lipton2003playing}, where e.g. \citet{babichenko2016empirical} show how to obtain an approximate equilibrium by sampling any mixed equilibrium. The PAC-learning framework proposed by \citet{valiant1984theory} has also been extended by \citet{blum2017collaborative}, who consider a collaborative game where players attempt to learn the same underlying prediction function, but each player has her own distribution over the space. In their work each player can sample from her own distribution, and the goal is to use information sharing among the players to reduce the sample complexity. Another line of work deals with extracting data from selfish agents for a variety of machine learning tasks \cite{charapodimata,dekel2008incentive,meir2012algorithms,Hardt2016}. We stress that these works are only parallel to this one, as in this paper there are no strategic agents, only myopic users who follow the ``best'' prediction they get. Moreover, in this work the strategic parties are those who provide the predictions, aiming at satisfying as many users as possible.

Our work is inspired by Dueling Algorithms \cite{immorlica2011dueling}. \citet{immorlica2011dueling} analyze an optimization problem from the perspective of competition, rather than from the point of view of a single optimizer. That work examines the dueling form of several optimization problems, e.g., the shortest path from the source vertex to the target vertex in a graph with random weights. While minimizing the expected length is a probable solution concept for a single optimizer (i.e., socially optimal), this is no longer the case in the defined duel. Our model is also related to Competing Bandits \cite{MansourSW18}.
\citet{MansourSW18} consider a competition between two bandit algorithms faced with the same sample, where users arrive one by one and choose between the two algorithms. This work is substantially different than \citet{MansourSW18}. To name a few differences, our work players also share the same sample but we consider an offline setting and not an online one;  infinite strategy spaces and not a finite set of actions; context in the form of property vector for each user; and an arbitrary number of asymmetric players, where asymmetry is reflected in the strategy space of each player.

Most relevant to our work is the work by \citet{porat2017best}. The authors present a learning task where a newcomer agent is given a sequence of examples, and wishes to learn a best-response to the players already on the market. They assume that the agent can sample triplets composed of instance, label and current market prediction, and define the agent's payoff as the proportion of points (associated with users) she predicts better than the other players. Indeed, \citet{porat2017best} introduce a learning task incorporating economic interpretation into the objective function of the (single) optimizer, but in fact does not provide any game-theoretic analysis. In contrast, this paper considers game-theoretic interaction between players, and its main contribution lies in the analysis of such interactions. %
Since learning dynamics consists of steps of unilateral deviations that improve the deviating player's payoff, the Best Response Regression of \citet{porat2017best} can be thought of as an initial step to this work.

\subsection{Our contribution}
Our contribution is three-fold. First, we explicitly suggest that prediction algorithms, like other products on the market, are in competition.
This novel view emphasizes the need for stability in prediction-based competition similar to  \citeauthor{Hotelling}'s stability in spatial competition \cite{Hotelling}.

Second, we introduce an extension of the PAC-learning framework for dealing with strategy profiles, each of which is a sequence of functions.  To illustrate why this is not an immediate application of the PAC-learning framework, recall that in the PAC-learning framework we assume a loss function which quantifies how good or bad a hypothesis is, regardless of the hypothesis space. In our games, however, a strategy profile cannot be determined to be good (being a PNE) or bad (there exist beneficial deviations) by itself; we must as well consider all profiles and in particular the possible beneficial deviations.  

We show a reduction from payoff maximization to a single-player loss minimization setting, which is later used to achieve bounds on the sample complexity for uniform convergence over the set of profiles. We also show that when players have approximate better-response oracles, they can learn an approximate PNE of the empirical game. The main technical contribution of this paper is an algorithm which, given $\epsilon,\delta$, samples a polynomial number of points in the game instance parameters, runs any $\epsilon$-better-response dynamics, and returns an $\epsilon$-PNE with probability of at least $1-\delta$.

Third, we consider games with at least one player with an infinite pseudo-dimension. We show a game instance where each player can learn the best prediction function from her hypothesis class if she were alone in the game, but a PNE of the empirical game is not generalized. This inability to learn emphasizes that strategic behavior can introduce further challenges to the machine learning community.

\section{Problem definition}
In this section we formalize the model. We begin with an informal introduction to elementary concepts in both game theory and learning theory that are used throughout the paper.

\paragraph{Game theory}
A non-cooperative game is composed of a set of players $\mN=\{1,\dots,N\}$; a strategy space $\mH_i$ for every player $i$; and a payoff function $\pi_i:\mH_1 \times \cdots \times \mH_N \rightarrow \mathbb R$ for every player $i$. The set $\mH =\mH_1 \times \cdots \times \mH_N$ contains all possible strategy profiles, and a tuple of strategies $\bl h =(h_1,\dots h_N) \in \mH$ is called a \textit{strategy profile}, or simply a profile. We denote by $\bl h_{-i}$ the vector obtained by omitting the $i$-th component of $\bl h$.

A strategy $h_i'\in \mH_i $ is called a \textit{better response} of player $i$ with respect to a strategy profile $\bl h$ if $\pi_i(h_i', \bl h_{-i}) > \pi_i(\bl h)$. Similarly, $h_i'$ is said to be an \textit{$\epsilon$-better response} of player $i$, w.r.t. a strategy profile $\bl h$ if $\pi_i(h_i', \bl h_{-i}) \geq  \pi_i(\bl h) +\epsilon$, and a \textit{best response} to $\bl h_{-i}$ if $\pi_i(h_i', \bl h_{-i}) \geq \sup_{h_i\in \mH_i}\pi_i(h_i,\bl h_{-i})$ .

We say that a strategy profile $\bl h$ is a \textit{pure Nash equilibrium} (herein denoted PNE) if every player plays a best response under $\bl h$. We say that a strategy profile $\bl h$ is an $\epsilon$-PNE if no player has an $\epsilon$-better response under $\bl h$, i.e., for every player $i$, it holds that $\pi_i(\bl h) \geq \sup_{h_i'\in \mH_i}\pi_i(h_i',\bl h_{-i})-\epsilon$.

\paragraph{Learning theory}

Let $\mX$ be a set, and let $F$ be a class of binary-valued functions $F\subseteq {\{0,1\}}^\mX$. Given a sequence $\mS=(x_1,\dots x_m)\in \mX^m$, we denote the \textit{restriction} of $F$ to $\mS$ by $F\cap \mS= \left\{ \left(f(x_1),\dots,f(x_m)\right)\mid f\in F  \right\}$. The \textit{growth function} of $F$, denoted $\Pi_F:\mathbb N \rightarrow \mathbb N$, is defined as $\Pi_F(m) = \max_{\mS\in \mX^m}\abs{F\cap \mS}$. We say that $F$ \textit{shatters} $\mS$ if $\abs{F\cap \mS}=2^{\abs{\mS}}$. The Vapnik-Chervonenkis dimension of a binary function class is the cardinality of the largest set of points in $\mX$ that can be shattered by $F$, 
$\vc(F)= \max\left\{m\in \mathbb N :\Pi_F(m)=2^m    \right\}$.

Let $H$ be a class of real-valued functions $H\subseteq \mathbb R^\mX$. The restriction of $H$ to $\mS \in \mX^m$ is analogously defined, $H\cap \mS= \left\{ \left(h(x_1),\dots,h(x_m)\right)\mid h\in H  \right\}$. We say that $H$ \textit{pseudo-shatters} $\mS$ if there exists $\bl r=(r_1,\dots,r_m)\in \mathbb R^m$ such that for every binary vector $\bl b=(b_1,\dots b_m)\in \{-1,1  \}^m$ there exists $h_{\bl b}\in H$ and for every $i\in [m]$ it holds that $\sign(h_{\bl b}(x_i)-r_i)=b_i $. The \textit{pseudo-dimension} of $H$ is the cardinality of the largest set of points in $\mX$ that can be pseudo-shattered by $H$, 
\[
\pdim(H)=\max\left\{m\in \mathbb N :\exists \mS\in\mX^m \text{ such that } \mS \text{ is pseudo-shattered by }  H \right\}.
\]

\subsection{Model}
\label{subsec:model}
We consider a set of users who are interested in a prediction provided by a set of competing players. Each user is associated with a vector $(x,y,t)$, where $x$ is the instance; $y$ is the label; and $t$ is the threshold that the user is willing to accept.

The players offer predictions to the users. When a user associated with a vector $(x,y,t)$ approaches player $i$, she produces a prediction $h_i(x)$. If $\abs{h_i(x)-y}$ is at most $t$, the user associated with $(x,y,t)$ will grant one monetary unit to player $i$. Otherwise, that user will move on to another player. We assume that users approach players according to the uniform distribution, although our model and results support any distribution over player orderings. Player $i$ has a set of possible strategies (prediction algorithms) $\mH_i$, from which she has to decide which one to use. Each player aims to maximize her expected payoff, and will act strategically to do so.

Formally, the game is a tuple $\langle \mZ,\mD,\mN,\mH,\pi \rangle$ such that
\begin{enumerate}
\item $\mZ$ is the examples domain $\mZ = \mX \times \mY \times \mT$, where $\mX\subset \R^n$ is the instance domain;  $\mY \subset \R$ is the label domain; and $\mT \subset \R_{\geq 0}$ is the tolerance domain.
\item $\mD$ is a probability distribution over $\mZ = \mX \times \mY \times \mT$.
\item $\mN$ is the set of players, where $\mN =[N]\defeq \{1,\dots,N\}$. A strategy of player $i$ is an element from $\mH_i \subseteq \mY ^ \mX$. The space of all strategy profiles is denoted by $\mH = \times_{i=1}^N \mH_i  $.

\item For $z=(x,y,t)$ and a function $g:\mX\rightarrow \mY$, we define the indicator $\mI(z,g)$ to be 1 if the distance between the value $g$ predicted for $x$ and $y$ is at most $t$. Formally,
\[
\mI(z,g)=
\begin{cases}
1 & \abs{g(x)-y} \leq t \\
0 & \text{otherwise}
\end{cases}.
\]
\item Given a strategy profile $\bl h=(h_1,\dots h_N)$ with $h_i \in \mH_i$ for $i\in \mN$ and $z=(x,y,t)\in \mZ$, let
\[
w_i(z;\bl h)=
\begin{cases}
0 & \text{ if } \mI(z,h_i)=0\\
\frac{1}{\sum_{i'=1}^N \mI(z,h_{i'})} & \text{otherwise}
\end{cases}.
\]
Note that $w_i(z;\bl h)$ represents the expected payoff of player $i$, w.r.t. the user associated with $z$. The payoff of player $i$ under $\bl h$ is the average sum over all users, and is defined by
\[
\pi_i (\bl h)= \E_{z\sim \mD}\left[ w_i(z;\bl h) \right].
\]
\item $\mD$ is unknown to the players.
\end{enumerate}

We assume players have access to a sequence of examples $\mS$, drawn i.i.d. from $\mD$. Given a game instance $\langle \mZ,\mD,\mN,\mH, \pi\rangle$ and a sample $\mS=\{z_1,\dots z_m\}$, we denote by $\langle \mZ,\mS \sim \mD^m ,\mN,\mH,\pi \rangle$ the \textit{empirical game}: the game over the same $\mN,\mH,\mZ$ and uniform distribution over the known $\mS \in \mZ^m$. We denote the payoff of player $i$ in the empirical game by 
\begin{equation}\label{eq:empirical payoff}
\pi_i^\mS (\bl h)= \E_{z \in \mS}\left[ w_i(z;\bl h) \right]=\frac{1}{m}\sum_{j=1}^m w_i(z_j;\bl h). 
\end{equation}
When $\mS$ is known from the context, we occasionally use the term \textit{empirical} PNE to denote a PNE of the empirical game. Since the empirical game is a complete information game, players can use the sample in order to optimize their payoffs. 

The optimization problem of finding a best response in our model is intriguing in its own right. Throughout the paper, we assume that each player $i$ has a polynomial $\epsilon$-better-response oracle. Namely, given a real number $\epsilon>0$, a strategy profile $\bl h$ and sample $\mS$, we assume that each player $i$ has an oracle that returns an $\epsilon$-better response to $\bl h_{-i}$ if such exists or answers false otherwise, which runs in time $\text{poly}(\frac{1}{\epsilon},m,N)$. Despite that this assumption may be non-trivial in general cases, we show in Section \ref{sec:bestresponsregression} such a best response oracle for the case of $\mH_i$ being the class of linear functions with constant input dimension (denoted by $n$ in the model above).  We also discuss situations where a better response cannot be computed efficiently in Section \ref{sec:discussion}, and present the applicability of our model for these cases as well.

\section{Meta algorithm and analysis}
\label{sec:meta}
Throughout this section we assume that the pseudo-dimension of $\mH_i$ is finite,  and we denote it by $d_i$, i.e., $\pdim(\mH_i)=d_i<\infty$. Our goal is to propose a generic method for finding an $\epsilon$-PNE efficiently. 
The method is composed of two steps: first, it attains a sample of ``sufficient'' size. Afterwards, it runs an $\epsilon$-better-response dynamics until convergence, and returns the obtained profile. The underlying idea is straightforward, but its analysis is non-trivial. In particular, we need to show two main claims: 
\begin{itemize}
\item Given a sufficiently large sample $\mS$, the payoff of each player $i$ in the empirical game is not too far away from her payoff in the actual game, with high probability. This holds concurrently for all possible strategy profiles.
\item An $\epsilon$-PNE exists in every empirical game. Therefore, players can reach an $\epsilon$-PNE of the empirical game quickly, using their $\epsilon$-better-response oracles.
\end{itemize}
These claims will be made explicit in forthcoming Subsections \ref{subsec:sample} and \ref{subsec:dyn}. We formalize the above discussion via Algorithm \ref{algorithm:betterres} in Subsection \ref{subsec:alg}.

\subsection{Uniform convergence in probability}
\label{subsec:sample}
We now bound the probability (over all choices of $\mS$) of having player $i$'s payoff (for an arbitrary $i\in \mN$) greater or less than its empirical counterpart by more than $\epsilon$. Notice that the restriction of $\mH_i$ to any arbitrary sample $\mS$, i.e., $\mH_i \cap \mS$, may be of infinite size. %
Nevertheless, the payoff under a strategy profile $\bl h \in \mH$  concerns the functions $w_1,\dots,w_N$ only and not the actual real-valued predictions produced by the players under $\bl h$; hence, uniform convergence of another set of functions should be argued, the class of functions that compose the payoff with the predictions.

More formally, given $z\in \mZ$ and $\bl h \in \mH$, let $w(z;\bl h)=(w_1(z;\bl h),\dots ,w_N(z;\bl h))$, and denote by $\mW$ the class of functions from $\mZ$ to $\{1,\frac{1}{2},\dots,\frac{1}{N},0\}^N$ such that
\begin{equation}
\label{eq:w not h}
\mW \defeq = \{ w(z;\bl h) \mid \bl h \in \mH \}.
\end{equation}
Notice that the set $\mW \cap \mS$ is finite. Moreover, for some $z\in \mZ$, $w(z;\bl h)=w(z;\bl h')$ may occur for $\bl h \neq \bl h'$; thus, to argue that uniform convergence holds, one would typically show that the growth function of $\mW$ is polynomial. However, analyzing $\mW$ directly seems tricky. To circumvent a direct analysis of $\mW$, we define auxiliary classes of binary functions $(\mF_i)_{i\in \mN}$. Let $\mF_i$ be a class of functions from $\mZ$ to $\{0,1\}$ such that 
\begin{equation}
\label{eq:defoff}
\mF_i \defeq \left\{\mI(z , h) \mid h\in \mH_i  \right\}.
\end{equation}
In the rest of this subsection, we prove properties of $\mF_i$ and then associate those properties with the payoffs of the game.

Observe that $\mF_i$ is a binary function class, and thus its complexity can be quantified using the VC dimension. On the other hand, we already know that the pseudo-dimension of $\mH_i$, which is a real-valued function class, is $d_i$. The following combinatorial lemma bounds the VC dimension of $\mF_i$ as a function of the pseudo-dimension of $\mH_i$. By doing so, we present an interesting connection between these two celebrated notations of expressive power of a space of functions (the VC dimension and the pseudo-dimension). 
\begin{lemma}
\label{LEMMA:REGTOCLASS}
$\vc(\mF_i) \leq 10d_i$.
\end{lemma}

After discovering the connection between the growth rate of $\mH_i$ and $\mF_i$, we can progress to bounding the growth of the product function class $\mF$ (which we will define shortly). For ease of notation, denote $\mI(z,\bl h)=(\mI(z,h_1),\dots ,\mI(z,h_N))$. Let $\mF$ be a class of functions from $\mZ$ to $\{0,1 \}^N$ defined by
\[
\mF \defeq \prod_{i=1}^N \mF_i = \left\{ \mI\left(z,\bl h  \right) \mid \bl h\in \mH\right\}.
\]
Note that every element in $\mF$ is a function from $\mZ$ to $\{0,1\}^N$. The restriction of $\mF$ to a sample $\mS$ is naturally defined by
\[
\mF \cap \mS = \prod_{i=1}^N (\mF_i\cap \mS) =  \left\{ \left(\mI(z_1,\bl h),\dots,\mI(z_m,\bl h)   \right) \mid \bl h \in \mH   \right\}.
\]
As a result,
\begin{equation}
\label{eq:size of f and s}
\abs{\mF \cap \mS}=\prod_{i=1}^N \abs{\mF_i \cap \mS}.
\end{equation}
We use Equation (\ref{eq:size of f and s}) to characterize the growth function of $\mF$, defined by $\Pi_\mF (m) = \max_{\mS\in \mZ^m    } \abs{ \mF \cap \mS   }$. We bound $\Pi_\mF (m)$  using Lemma \ref{LEMMA:REGTOCLASS} and the Sauer-Shelah lemma.
\begin{lemma}
\label{lemma:growthsumvc}
$\Pi_\mF(m) \leq (em)^{10\sum_{i=1}^N d_i}$.
\end{lemma}

Before we claim for uniform convergence in probability, we must relate the number of distinct profiles (i.e., elements in $\mW \cap \mS$) in the empirical game under $\mS$ and the size of $\mF \cap \mS$. The following claim shows that the size of $\mF \cap \mS$ is an upper bound on the size of $\mW \cap \mS$.
\begin{claim}
\label{claim:size and f}
It holds that $\Pi_{\mW}(m) \leq \Pi_{\mF}(m)$.
\end{claim} 
Next, we bound the probability of a player $i$'s payoff being ``too far'' from its empirical counterpart. The proof of Lemma \ref{lemma:uniconvergenceoneplayer} below goes along the path of \citet{vapnik2015uniform}. Since in our case $\mF$ is not a binary function class, a few modifications are needed. 
\begin{lemma} 
\label{lemma:uniconvergenceoneplayer}
Let $m$ be a positive integer, and let $\epsilon>0$. It holds that
\[
\Pr_{\mS\sim \mD^m}\left(\exists \bl{h} : \abs{\pi_i(\bl h)-\pi_i^{\mS}(\bl h)} \geq  \epsilon \right) \leq 4 \Pi_{\mW}(2m)e^{-\frac{\epsilon^2 m}{8}}.
\]
\end{lemma}
The following theorem bounds the probability that any player $i$ has a difference greater than $\epsilon$ between its payoff and its empirical payoff (over the selection of a sample $\mS$), uniformly over all possible strategy profiles. This is done by simply applying the union bound on the bound already obtained in Lemma \ref{lemma:uniconvergenceoneplayer}.
\begin{theorem}
\label{thm:unionbound}
Let $m$ be a positive integer, and let $\epsilon>0$. It holds that
\begin{equation}
\label{eq:inthmunionbound}
\Pr_{\mS\sim \mD^m}\left(\exists i\in \mN : \sup_{\bl h \in \mH}\abs{\pi_i(\bl h)-\pi_i^{\mS}(\bl h)} \geq  \epsilon \right) \leq 4N (2em)^{10\sum_{i=1}^N d_i}e^{-\frac{\epsilon^2 m}{8}}.
\end{equation}
\end{theorem}

\subsection{Existence of a PNE in empirical games}
\label{subsec:dyn}
In the previous subsection we bounded the probability of a payoff vector being too far from its counterpart in the empirical game. Notice, however, that this result implies nothing about the existence of a PNE or an approximate PNE: for a fixed $\mS$, even if $\sup_{\bl h \in \mH}\abs{\pi_i(\bl h)-\pi_i^{\mS}(\bl h)} < \epsilon$ holds for every $i$, a player may still have a beneficial deviation. Therefore, the results of the previous subsection are only meaningful if we show that there exists a PNE in the empirical game, which is the goal of this subsection.
We prove this existence using the notion of {\em potential games} \cite{monderer1996potential}.

A non-cooperative game is called a potential game if there exists a function $\Phi:\mH \rightarrow\mathbb R$ such that for every strategy profile $\bl h=(h_1,\dots,h_N) \in \mH$ and every $i\in \mN$, whenever player $i$ switches from $h_i$ to a strategy $h_i'\in \mH_i$, the change in her payoff function equals the change in the potential function, i.e.,
\[
\Phi (h'_{{i}},\bl h_{{-i}})-\Phi (h_{{i}},\bl h_{{-i}})=\pi_{{i}}(h'_{{i}},\bl h_{{-i}})-\pi_{{i}}(h_{{i}},\bl h_{{-i}}).
\]
\begin{theorem}[\cite{monderer1996potential,rosenthal1973class}] 
\label{thm:pot}
Every potential game with a finite strategy space possesses at least one PNE.
\end{theorem}
Obviously, in our setting the strategy space of a game instance $\langle \mZ,\mD,\mN,\mH ,\pi \rangle$ is typically infinite.
Infinite potential games may also possess a PNE (as discussed in \cite{monderer1996potential}), but in our case the distribution $\mD$ is approximated from samples and the empirical game is finite, so no stronger claims are needed. Lemma \ref{lemma:PNEexistence} below shows that every empirical game is a potential game.
\begin{lemma}
\label{lemma:PNEexistence}
Every empirical game  $\langle \mZ,\mS \sim \mD^m ,\mN,\mH ,\pi \rangle$ has a potential function.
\end{lemma}
As an immediate result of Theorem \ref{thm:pot} and Lemma \ref{lemma:PNEexistence}, 
\begin{corollary}
\label{corollary:PNEexistence}
Every empirical game  $\langle \mZ,\mS \sim \mD^m ,\mN,\mH,\pi \rangle$ possesses at least one PNE.
\end{corollary}

After establishing the existence of a PNE in the empirical game, we are interested in the rate with which it can be ``learnt''. More formally, we are interested in the convergence rate of the dynamics between the players, where at every step one player deviates to one of her $\epsilon$-better responses. Such dynamics do not necessarily converge in general games, but do converge in potential games. By examining the specific potential function in our class of (empirical) games, we can also bound the number of steps until convergence.

\begin{lemma}
\label{lemma:betterconverge}
Let  $\langle \mZ,\mS \sim \mD^m ,\mN,\mH,\pi \rangle$ be any empirical game instance. After at most $O\left(\frac{\log N}{\epsilon}  \right)$ iterations of any $\epsilon$-better-response dynamics, an $\epsilon$-PNE of the empirical game is obtained.
\end{lemma}
\textit{Remark:} Due to Equation (\ref{eq:empirical payoff}), we know that any beneficial deviation improves the payoff of the deviating player by at least $\frac{1}{mN}$; hence, Lemma \ref{lemma:betterconverge} suggests that an (exact) empirical PNE could be obtained using at most $O\left(mN\log N  \right)$ iterations of any $\frac{1}{mN}$-better-response dynamics.

\subsection{Learning $\epsilon$-PNE with high probability}
\label{subsec:alg}
In this subsection we leverage the results of the previous Subsections \ref{subsec:sample} and \ref{subsec:dyn} to devise Algorithm \ref{algorithm:betterres}, which runs in polynomial time and returns an approximate equilibrium with high probability. More precisely, we show that Algorithm \ref{algorithm:betterres} returns an $\epsilon$-PNE with probability of at least $1-\delta$, and has time complexity of $\text{poly}\left(\frac{1}{\epsilon},m,N,\log\left( \frac{1}{\delta}\right),d \right)$. % 
As in the previous subsections, we denote $d=\sum_{i=1}^N d_i $.

First, we bound the required sample size. Using standard algebraic manipulations on Equation (\ref{eq:inthmunionbound}), we obtain the following.
\begin{lemma}
\label{lemma:deltaandm}
Let $\epsilon,\delta \in (0,1)$, and let
\begin{equation}
\label{eq:lemma:deltaandm}
m \geq \frac{320d}{\epsilon^2} \log\left( \frac{160d}{\epsilon^2} \right) +\frac{160d\log(2e)}{\epsilon^2}+\frac{16}{\epsilon^2}\log\left( \frac{4N}{\delta} \right).
\end{equation}
With probability of at least $1-\delta$ over all possible samples $\mS$ of size $m$, it holds that
\[
\forall i\in \mN : \sup_{\bl h \in \mH}\abs{\pi_i(\bl h)-\pi_i^{\mS}(\bl h)} <  \epsilon.
\]
\end{lemma}
Given $\epsilon,\delta$, we denote by $m_{\epsilon,\delta}$ the minimal integer $m$ satisfying Equation (\ref{eq:lemma:deltaandm}). Lemma \ref{lemma:deltaandm} shows that $m_{\epsilon,\delta}=O\left(\frac{d}{\epsilon^2}\log\left( \frac{d}{\epsilon^2} \right) + \frac{1}{\epsilon^2}\log\left( \frac{N}{\delta} \right) \right)$ are enough samples to have all empirical payoff vectors $\epsilon$-close to their theoretic counterpart coordinate-wise (i.e., in the $L^\infty$ norm), with a probability of at least $1-\delta$. 

Next, we bind an approximate PNE in the empirical game with an approximate PNE in the (actual) game.
\begin{lemma}
\label{lemma:empiseq}
Let  $m \geq m_{\frac{\epsilon}{4},\delta}$ and let $\bl h$ be an $\frac{\epsilon}{2}$-PNE in  $\langle \mZ,\mS \sim \mD^m ,\mN,\mH,\pi \rangle$. Then $\bl h$ is an $\epsilon$-PNE with probability of at least $1-\delta$. 
\end{lemma}

Recall that Lemma \ref{lemma:betterconverge} ensures that every $O\left(\frac{\log N}{\epsilon} \right)$ iterations of any $\epsilon$-better-response dynamics must converge to an $\epsilon$-PNE of the empirical game. In each such iteration, a player calls her approximate better-response oracle, which is assumed to run in $\text{poly}(\frac{1}{\epsilon},m,N)$ time. Altogether, given $\epsilon$ and $\delta$, Algorithm \ref{algorithm:betterres} runs in $\text{poly}\left(\frac{1}{\epsilon},N,\log\left( \frac{1}{\delta}\right),d \right)$ time, and returns an $\epsilon$-PNE with probability of at least $1-\delta$.
\begin{corollary}
Let $\langle \mZ,\mD ,\mN,\mH,\pi \rangle$ be a game, $\epsilon,\delta\in (0,1)$, and let $\bl h$ be the output of Algorithm \ref{algorithm:betterres}. With probability of at least $1-\delta$, $\bl h$ is an $\epsilon$-PNE.
\end{corollary}

\begin{algorithm}[t]
\DontPrintSemicolon
\caption{ Approximate PNE w.h.p. via better-response dynamics\label{algorithm:betterres}}
\KwIn{$\delta, \epsilon \in (0,1)$
}
\KwOut{a strategy profile $\bl h$}
set $m = m_{\frac{\epsilon}{4},\delta}$ \tcp*{the minimal integer $m$ satisfying Equation (\ref{eq:lemma:deltaandm})}
sample $\mS$ from $\mD^m$\;
execute any $\frac{\epsilon}{2}$-better-response dynamics on the empirical game corresponding to $\mS$ until convergence, and obtain a strategy profile $\bl h$ that is an empirical $\frac{\epsilon}{2}$-PNE \label{algorithm:betterres:dynamics}\;
\Return{$\bl h$}
\end{algorithm}
\section{Linear Strategy Space}
\label{sec:bestresponsregression}
Linear regression is extensively studied in the statistics/machine learning literature, and hence we find it the most appropriate use-case to examine in our competitive setting. In this section, we shall assume $\mH_i$ is the linear strategy space for an arbitrary  $i\in \mN$. Formally, $\mH_i$ is the function class of all linear mappings from $\mX\subseteq\R^n$ to $\R$. For ease of notation, we shall treat every $h_i\in \mH_i$ as a vector in $\R^n$ that corresponds to the mapping $x \mapsto \inprod{h_i}{x}$ for $x\in \R^n$. \footnote{We shall keep using non-bold notation for single strategies and input vectors, to distinguish between those and strategy profiles.} Under this representation, the class $\mH_i$ can be referred to as $\mathbb R^{n}$. We shall further assume that the dimension of the input $n$ is constant.

We first present an oracle for the best linear response. We describe informally how a player $i$ can compute a best response against any $\bl h_{-i}$, and then devise an algorithm that is inspired by the one proposed in \citet{porat2017best}. Algorithm \ref{alg:blr} finds a best linear response against \textit{any} strategies selected by the other players, not necessarily linear strategies, and would apply even if the other players employ more sophisticated algorithms, e.g., neural networks. Algorithm \ref{alg:blr} runs in polynomial time in the size of the sample $m$. Afterwards, we apply Algorithm \ref{algorithm:betterres} along with the best response oracle implemented in Algorithm \ref{alg:blr} on several synthetic datasets in a two-player game with linear strategies, and examine the equilibrium structure and properties.

\subsection{A Polynomial Best-Response Oracle}

Before we present the formal algorithm, we give a high-level intuition of how it works. Let $\mS=(z_j)_{j=1}^m$ be an arbitrary sample, $i\in \mN$ be an arbitrary player index, and $\bl h_{-i}$ be arbitrary strategies of all players but $i$. Denote by $M$ the mapping $h_i \mapsto \bl w$ such that 
\[
M(h_i) = \left(  w_i(z_1,\bl h_{-i},h_i),\dots,w_i(z_m,\bl h_{-i},h_i)\right).
\]
Observe that $M$ need not be onto, as not all vectors in $\{0,\frac{1}{N},\frac{2}{N},\dots 1\}^m$ may be in the image of $M$, nor one-to-one, as $M(h_i)=M(h_i')$ for $h_i,h_i'\in \mH_i, h_i\neq h_i'$ may occur.
In addition, the payoff of player $i$ under $h_i$ is precisely $\norm{M(h_i)}_1$. 
While the size of the target set of $M$ is exponential in $m$, Lemma \ref{LEMMA:REGTOCLASS} implies that the size of the image of $M$ is only polynomial in $m$. The idea is therefore to do an exhaustive search over the image of $M$ in polynomial time, find a vector $\bl w^* \in \argmax_{\bl w \in \text{Image}(M)} \norm{\bl w}_1$, and then find $h_i^*\in \mH_i$ such that $M(h_i^*)=\bl w^*$. 

Next, we introduce an auxiliary decision problem that facilitates the search over the image of $M$,  the Partial Vector Feasibility problem (PVF). For every strategy $h_i\in H_i$ and a point $(x,y,t)$, we focus on the term $\abs{\inprod{h_i}{x}-y}$. In case it is less or equal to $t$, player $i$ will get (some fraction) of this point. Otherwise, she will not, and two cases can occur: either $\inprod{h_i}{x}-y$ is greater than $t$ (above $t$), or $\inprod{h_i}{x}-y$ is less than $-t$ (below $t$). A PVF problem is a decision problem that asks whether a partial set of the points can be classified into a particular assignment of these three cases; it is formally defined in Algorithm \ref{alg:pvf}.

\begin{algorithm}[t]
\caption{\textsc{Partial Vector Feasibility (PVF)} \label{alg:pvf}}
\KwIn{a sequence of examples $\mS=(x_j,y_j,t_j)_{j=1}^m$, and a vector $\bl v\in \{1,a,b,0\}^m $.	
}
\KwOut{a strategy $h_i \in \R^n$ satisfying\\
	\begin{itemize}
	\item if $v_j=1$, then $\abs{ \inprod{h_i}{x_j} -y_j}\leq t_j$ \tcp{$\mI(z_j,h_i)=1$}
	\item if $v_j=a$, then $\inprod{h_i}{x_j} -y_j>t_j$ \tcp{above}
	\item if $v_j=b$, then $\inprod{h_i}{x_j} -y_j<-t_j$ \tcp{below}
	\item if $v_j=0$, there is no constraint for the $j$'th point
	\end{itemize}
	if such exists, and $\phi$ otherwise.
}
\end{algorithm}
Note that PVF is solvable in polynomial time via linear programming. We are now ready to present the Best Linear Response (BLR) algorithm. BLR has three main steps:
\begin{enumerate}
\item Compute the possible gain from each point in the sample. The possible gain from each point is a function of $\bl h_{-i}$.
\item Find all feasible subsets of points player $i$ can satisfy concurrently (i.e., vectors in $\mF_i \cap \mS$, where $\mF_i$ is as defined in Equation (\ref{eq:defoff})).
\item Return a strategy that achieves the highest possible payoff.
\end{enumerate}

The first step consists of a straightforward computation. To motivate the second step, notice that if $\mI(z_j,h_i)=0$, then either $\inprod{h_i}{x_j} -y_j>t_j$ or $\inprod{h_i}{x_j} -y_j <-t_j$ holds. Therefore, we identify all vectors $\bl v=(v_1,\dots v_m)\in \{1,a,b\}^m$ such that there exists $h_i \in \mH_i$ and $v_j = 1$ if $\mI(z_j,h_i)=1$; $v_j=a$ if $\inprod{h_i}{x_j} -y_j>t_j$ (``above''); and $v_j = b$ if $\inprod{h_i}{x_j} -y_j<-t_j$ (``below''). This is done by recursively partitioning $\{1,a,b\}^m$, where in each iteration we consider only a prefix of the entries, while the suffix is masked with ``0'' (see $\textsc{PVF}$). 
At the end of this step, we have fully identified the set $\mF_i \cap \mS$, since every vector in $\{1,a,b\}^m$ can be mapped to a vector in $\mF_i \cap \mS$ by replacing  $a,b$ with $0$. Due to the bijection between $\mF_i \cap \mS$ and the image of $M$ (see Subsection \ref{subsec:sample}), we have essentially discovered the image of $M$; thus, we can pick a feasible vector that corresponds to the highest payoff. Finally, we find a strategy that attains the highest payoff by invoking \textsc{PVF} for the last time. The above discussion is formulated via Algorithm \ref{alg:blr}.

\begin{algorithm}[t]
\DontPrintSemicolon
\KwIn{$\mS=(x_j,y_j,t_j)_{j=1}^m$, $\bl h_{-i}$}
\KwOut{A best response to $\bl h_{-i}$}
for every $j\in[m]$, $w_j \gets \frac{1}{\sum_{i'\neq i} \mI(z_j,h_{i'})+1} $ \label{alg:rnwithoracle:linew} \tcp*{player $i$ can get up to $w_j$ for $z_j$}
$\boldsymbol {v} \gets \{0 \}^m$  \tcp*{ $\boldsymbol {v} =(v_1,v_2,\dots,v_m)$}
$\mR_0 \gets \left\{\boldsymbol {v} \right\}$  \;
\For {$j=1$ to $m$} { \label{alg:rnwithoracle:lineifor}
	$\mR_j \gets \emptyset$ \;
	\For {$\boldsymbol {v} \in \mR_{j-1}$} { \label{alg:rnwithoracle:linefor}
		\For {$\alpha \in \{1,a,b\}$} { \label{alg:rnwithoracle:lineforallv}
			\If {$\textsc{PVF} \left(\mS,(\boldsymbol {v}_{-j},\alpha )\right) \neq \phi$}{ \label{alg:rnwithoracle:lineoracle}
			add $(\boldsymbol {v}_{-j},\alpha )$ to $\mR_j$  \tcp*{$(\boldsymbol {v}_{-j},\alpha )=(v_1,\dots v_{j-1},\alpha,v_{j+1},\dots,v_m)$} 
			}
		}
	}
}
$\boldsymbol {v}^* \gets  \argmax_{\boldsymbol {v} \in \mR_m} \sum_{j=1}^m w_j\ind_{v_j=1}$ \label{alg:rnwithoracle:argmax}\;
\Return \textsc{PVF}$(\mS,\bl v^*)$ \; \label{alg:rnwithoracle:return}
\caption{\textsc{Best Linear Response} (BLR) \label{alg:blr}}
\end{algorithm}

Recall that we assume that the input space $n$ is constant. 
\citet{porat2017best} show that the second step (the for loop in line \ref{alg:rnwithoracle:lineifor}) is done in time $\text{poly}(m)$, and that $\mR_m$ is of polynomial size. The first step (line \ref{alg:rnwithoracle:linew}) and the last step (lines \ref{alg:rnwithoracle:argmax} and \ref{alg:rnwithoracle:return}) are clearly executed in polynomial time. Overall, Algorithm \ref{alg:blr} runs in polynomial time. In addition, since it considers all possible distinct strategies using $\mF_i \cap \mS$ and takes the one with the highest payoff, it indeed returns the best linear response with respect to $\bl h_{-i}$ in the empirical game. To sum,
\begin{theorem}
\label{thm:best linear response}
Let $\mH_i$ be the linear strategy space, $\mH_{-i}$ be the product of any strategy spaces, $\bl h_{-i}$ be an element from $\mH_{-i}$  and let $\mS$ be a sample of size $m$. Algorithm \ref{alg:blr} finds $h_i^*$ such that $h_i^* \in \argmax_{h_i\in \mH_i} \pi_i^\mS (\bl h)$ in time $\text{poly}(m)$.
\end{theorem}
We note that the best response need not be unique, as Algorithm \ref{alg:blr} employs two tie breakers. The first tie breaker is the selection of a feasible vector $\bl v^*$ with the highest payoff in Line \ref{alg:rnwithoracle:argmax}, since several such vectors may exist. Second, the call $\textsc{PVF}(\mS,\bl v^*)$ in Line \ref{alg:rnwithoracle:return} returns a single strategy, albeit the linear program it solves probably possesses infinitely many optimal solutions.

\subsection{Simulations with Synthetic Data in $\R^2$}
\label{SUBSEC:SIM}
The goal of this subsection is to visualize equilibrium profiles in several scenarios, and analyze the implications of different tolerance levels on the equilibrium structure. In service of that, we focus on two-player games with $\mX \times \mY=\R^2$, and the linear strategy space for both players. Recall that an equilibrium profile can be obtained by executing Algorithm \ref{algorithm:betterres}, where Algorithm \ref{alg:blr} plays the role of a best response (which is a special case of $\epsilon$-better response). Since our main purpose is to discuss the structure of equilibria, we focus solely on the empirical games and completely neglect the generalization analysis. Note, however, that the results obtained in the previous sections suggest that the payoffs in equilibrium will be similar to their empirical counterparts.

We now describe the high-level details of the synthetic datasets we employ. We let the instance $x$ be uniformly distributed in a close segment. In addition, we consider four conditions for the relations between the instance $x$ and the value $y$:
\begin{enumerate}
\item Linear: for every instance $x$, its value $y$ is a linear function of $x$, plus an additive normal noise. Namely, $y=ax+b+\epsilon$ for $a,b\in \R$ and $\epsilon\sim Normal(0,1)$. \footnote{$a$ and $b$ are not related to the arbitrary numbers used in Algorithm \ref{alg:pvf}.}
\item V-shape: the value of every $x$ is determined by $y=a\abs{x-x_0}+b+\epsilon$ for $x_0,a,b\in \R$ and $\epsilon\sim Normal(0,1)$.
\item X-shape: for every instance $x$, the value $y$ is distributed as follows:
\[
y=
\begin{cases}
a_1x+b_1+\epsilon & \text{w.p. }\frac{1}{2}\\
a_2x+b_2+\epsilon & \text{w.p. }\frac{1}{2}
\end{cases},
\]
where $a_1,b_1,a_2,b_2 \in \R$ and $\epsilon\sim Normal(0,1)$.
\item Piecewise: for every $x$, the value $y$ is given by
\[
y=
\begin{cases}
a_1x+b_1+\epsilon & \text{if }x\leq x_0\\
a_2x+b_2+\epsilon & \text{else}
\end{cases},
\]
where $x_0,a_1,b_1,a_2,b_2 \in \R$ and $\epsilon\sim Normal(0,1)$.
\end{enumerate}
\begin{figure*}[t]
\centering
\includegraphics[scale=0.6]{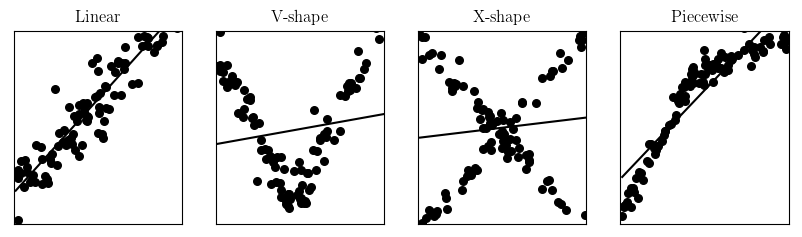}
\caption{Visualization of the synthetic datasets we employ. Each sub-figure represents one condition, and the points are realized values drawn from the distribution defined on $\mX \times \mY$ per condition. The line is the linear least squares, a canonical prediction algorithm in statistics/machine learning.
\label{fig:mse-two-dim}}
\end{figure*}
We have so far described the distribution over $\mX \times \mY$ for all four conditions. A visualization of generated datasets according to these conditions appears in Figure \ref{fig:mse-two-dim}. To complete a game description, we need to specify the tolerance levels. Different tolerance levels alter the equilibrium behavior, and thus are interesting to explore. In these simulations, we use a constant tolerance level, i.e., given a set $(x_j,y_j)_{j=1}^m$ associated with one of the four conditions above and a value $t\in \R$, the dataset we analyze is $(x_j,y_j,t)_{j=1}^m$. By doing so, we do not exploit the full generality of the model and algorithms. Nevertheless, since the tolerance levels are non-obvious to simulate, we feel that this is justified. To examine the extent to which the payoffs are sensitive to the tolerance levels, we consider three levels of tolerance, which vary across the four conditions above.

We generated\footnote{The specifics are elaborated in {\ifnum\Includeappendix=1{Section \ref{sec:simulation specifics}}\else{in the appendix}\fi}. Our code is available at:  \url{https://github.com/omerbp/Regression-Equilibrium}.} 12 datasets, composed of the four conditions and three tolerance levels for each condition. For each dataset, we run Algorithm \ref{algorithm:betterres}, with Algorithm \ref{alg:blr} playing the role of a best response oracle, until convergence and observed the obtained equilibrium profile.

Our findings are reported in Figure \ref{fig:two-dim}.
Each column of sub-figures corresponds to one of the four conditions above. The rows correspond to tolerance levels in a decreasing order. For instance, the top leftmost sub-figure represents the linear condition with high tolerance level. The colored lines are the strategies of the players in a PNE profile, as obtained by Algorithm \ref{algorithm:betterres}. A point is circled-red (square-blue) if only the dashed-red (solid-blue) player produces an accurate prediction for it. Alternatively, a point is star-green if both players produce an accurate prediction for it, and x-black if both predictions are inaccurate (i.e., point $j$ is black $\abs{h_i(x_j)-y_j}>t$ for $i\in \{1,2\}$ and the corresponding tolerance level $t$). 

Observe that the players' strategies are distinguishable, even for the linear condition. 
Noticeably, the number of star-green points decreases with the tolerance level. It symbolizes that the number of points obtaining an accurate prediction from both players decreases with the tolerance. In contrast, the number of x-black points increases with the tolerance, as less points get an accurate prediction when the tolerance is low. Moreover, notice that a high tolerance allows the players to gain almost all points, but when the tolerance is low roughly half of the points obtain an inaccurate prediction. As for player payoffs, the two players get almost the same payoff in all datasets, albeit theoretically the payoffs may greatly differ; this is due to the symmetric nature of the datasets.
 
\begin{figure*}[t]
\centering
\includegraphics[scale=0.65]{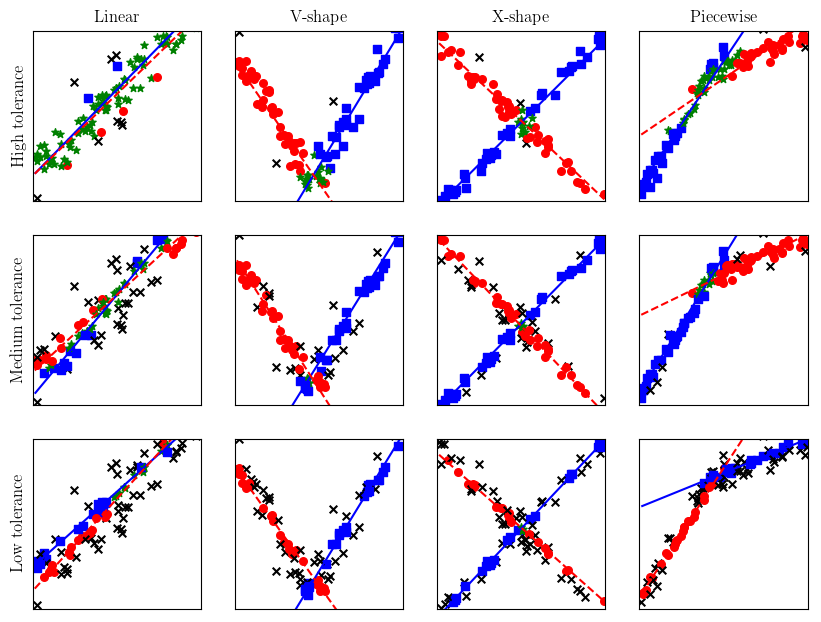}
\caption{Visualization of equilibrium profiles per condition and tolerance level. Each column of sub-figures is associated with one of the four conditions, and each row with one tolerance level. For example, the bottom rightmost sub-figure corresponds to the piecewise condition with low tolerance level. The players' strategies are the lines, and each point in the dataset is colored and patterned according to the player who obtains it. Namely, a point is circle-red (square-blue) if only the dashed-red (solid-blue) player predicts it accurately, in terms of the tolerance. A point is star-green if both players predicted it accurately, and x-black if both predictions are inaccurate. 
\label{fig:two-dim}}
\end{figure*}

\section{Extensions}
\label{sec:extensions}
In this section we answer two interesting questions. First, in Subsection \ref{subsec:exampleandinf}, we deal with an infinite pseudo-dimension. While we show learning may not occur, similarly to other machine learning applications (see, e.g., \cite{shalev2014understanding}), we obtain this result even if the players fully observe $\mD$. This conclusion is in sheer contrast to classical results, since when $\mD$ is observed there is essentially nothing to learn. In our setting, however, equilibrium strategy is another object that has to be learned. Second, we analyze a variant of the model that is aligned with the payoff function of \citet{porat2017best} and \citet{immorlica2011dueling}, Namely, where each user (associated with an example) grants one monetary unit to the player who provides the most accurate prediction. Unlike the existence we proved in Subsection \ref{subsec:dyn} for the main model, under this variant an empirical PNE may not exist.

\subsection{Learnability in games with infinite dimension}
\label{subsec:exampleandinf}

While Lemma \ref{LEMMA:REGTOCLASS} upper bounds the VC dimension of  $\mF_i$, the following Claim \ref{claim:flowerbound} puts a lower bound on it. Claim \ref{claim:flowerbound} implies that if $\pdim(\mH_i)$ is infinite, so is $\vc(\mF_i)$.
\begin{claim} 
\label{claim:flowerbound}
$\vc(\mF_i) \geq \pdim(\mH_i)$.
\end{claim}

Classical results in learning theory suggest that if $\vc(\mF_i)=\infty$, a best response on the sample may not generalize to an approximate best response w.h.p. To see this, imagine a ``game'' with one player, who seeks to maximize her payoff function. No Free Lunch Theorems (see, e.g., \cite{wolpert1997no}) imply that with a constant probability, the player cannot get her payoff within a constant distance from the optimal payoff. We conclude that in general games, if a player has a strategy space with an infinite pseudo-dimension, she may not be able to learn. However, in the presence of such a player, can other players with a finite pseudo-dimension learn an approximate best-response? In our setting, players are interacting with each other, and player payoffs are a function of the whole strategy profile; thus, different challenges may arise.

The strategy space of a player $i$ in an empirical game induced by a sample $\mS$ is $\mH_i \cap \mS$, namely the restriction of her strategy space to the sample. Due to the structure of the payoff function, player $i$ can compute an empirical best response to $\bl h_{-i}\cap \mS$ without knowing $\bl h_{-i}$ nor $\bl \mH_{-i}$. As the sample size increases $\mS$ becomes a representative of $\mD$, and we  employed uniform convergence to conclude that w.h.p. $\mH \cap \mS$ is ``similar'' to $\mH $; hence, players use the sample to learn both the underline distribution and a best response to the other players. In the rest of this subsection we give an example where knowledge about the strategy spaces of the other players is crucial. In particular, we show that if player $i$ (for $i=1$) has a strategy space with infinite dimension, she can trick the other players and drive them to her preferable outcome. 
\begin{example}
\label{example:unlearnability} Let $\mD$ be a density function over $\mZ=[0,2] \times [0,1] \times \left\{  \frac{1}{2}\right\}$ as follows: 
\[
\mD(x,y,t) = 
\begin{cases}
\frac{1}{2} & 0\leq x<1,y=0, t=\frac{1}{2}\\
\frac{1}{2} & 1\leq x\leq 2,y=1, t=\frac{1}{2}\\
0 & \text{otherwise}
\end{cases}.
\]

We now define several strategies that will constitute the strategy spaces. Let $h^0(x)\equiv 0,h^1(x) \equiv 1$ be constant functions mapping each instance $x$ to 0, 1 correspondingly. Moreover, for any $\mS\subset \mZ$, denote
\[
h^{\mS\rightarrow 0}(x) = 
\begin{cases}
0 & \exists y,t: (x,y,t)\in \mS \\
\ind_{1\leq x\leq 2} & \forall y,t: (x,y,t)\notin   \mS 
\end{cases}, \quad
h^{\mS\rightarrow 1}(x) = 
\begin{cases}
1 & \exists y,t: (x,y,t)\in \mS \\
\ind_{1\leq x\leq 2} & \forall y,t: (x,y,t)\notin  \mS 
\end{cases}.
\]
In other words, $h^{\mS\rightarrow 0}$ labels 0  (and $h^{\mS\rightarrow 1}$ labels 1)  every instance $x$ that appears in the projection of $\mS$ on its first entry; otherwise, it labels it correctly according to $\mD$, i.e., $h^{\mS\rightarrow 0}(x)=\ind_{1\leq x\leq 2}$ (and similarly, $h^{\mS\rightarrow 1}(x)=\ind_{1\leq x\leq 2}$). Let $\mN=\{1,2,3\}$, $\mH_2=\mH_3=\{h^0,h^1\}$, and denote
\[
\mH_1 = \{h^{\mS\rightarrow 0} \mid \mS \subset \mZ   \}   \cup \{h^{\mS\rightarrow 1} \mid \mS \subset \mZ   \} \cup \{h^0,h^1\}.
\]
Notice that $\mH_1$ is ``complex'', i.e., $\pdim(\mH_1)=\infty$. Next, we define a game $\mG$ such that $\mG=\langle \mZ,\mD ,\mN,\mH,\pi \rangle$ for $\mH=\mH_1\times \mH_2 \times \mH_3$. Denote by $\bl h$ the strategy profile $\bl h=(h^{\mS\rightarrow 0},h^1,h^1)$, and by $\bl h'$ the strategy profile $\bl h'=(h^0,h^1,h^1)$. Observe that $\pi_1(\bl h)=\frac{2}{3}, \pi_2(\bl h)=\pi_3(\bl h)=\frac{1}{6}$, while $\pi_1(\bl h')=\frac{1}{2}, \pi_2(\bl h')=\pi_3(\bl h')=\frac{1}{4}$. Moreover, $h^{\mS\rightarrow 0}$ is a best response of player 1 to $\bl h_{-1}$, while $h^1$ is a best response of player 2 to $\bl h'_{-2}$ (similarly for player 3). As a result, player 1 would like to play $h^{\mS\rightarrow 0}$ against $\bl h_{-1}$, but let player 2 and 3 think she plays $h^0$ and therefore mislead them to assume that playing $h^1$ (as they do under $\bl h'$) is a best response.

Since $h^0$ and $h^{\mS\rightarrow 0}$ coincide on every sample $\mS$, players 2 and 3 cannot distinguish between $\bl h$ and $\bl h'$ by observing the empirical payoffs. Moreover, with probability of at least $\frac{1}{4}$ over all choices of $\mS$ for $\abs{\mS} \geq 15$ we have $\frac{1}{2}<\frac{1}{m}\sum_{j=1}^m  \ind_{y_j=1} < \frac{3}{4}$ (see{\ifnum\Includeappendix=1{ Claim \ref{claim:auxhoeffbin} in}\fi} the appendix). Under such a sample $\mS$, player 2 and 3 indeed play a best response under $\bl h'$ in the empirical game induced by $\mS$ (while player 1 is not). Consequently, their lack of knowledge about $\mH_1$ forbids them to better-respond to player 1 and leads them to sub-optimal payoffs (each gets $\frac{1}{6}$ compared to $\frac{1}{4}$ if she deviates) with constant probability. By playing $h^{\mS\rightarrow 0}$, which is sub-optimal in the empirical game of $\mG$ on such $\mS$ (since $\pi_1^\mS(h^{\emptyset \rightarrow 0}, \bl h_{-1})-\pi_1^\mS(\bl h)>\frac{1}{3}$), player 1 can manipulate the others to think that she plays $h^0$.
\end{example}

The intuition behind the example is as follows. Despite $\mD$ being relatively simple, the strategy selected by player $1$ is ``complex'', in the sense that its behavior on the sample and on the population is substantially different; thus, players $2$ and $3$ are only playing a best response to the observable, simple way player $1$ is playing on the sample, yet misinterpret the extension of player $1$'s strategy to the population. The results of the previous sections imply that this phenomena cannot occur when the pseudo dimensions are finite, as long as the sample is large enough. Another interesting point is that in Example \ref{example:unlearnability} each player can find a strategy that maximizes her payoff if she were alone using a small number of samples. Indeed, this inability to generalize from samples follows solely from strategic behavior. Notice that if player 2 has knowledge of $\mH_1$, she can infer that her strategy under $\bl h'$ is sub-optimal. However, knowledge of the strategy spaces of other players is a heavy assumption: the better-response dynamics we discussed in Subsection \ref{subsec:dyn} only assumed that each player can compute a better response.

\subsection{Non-existence of PNE in the Direct Attraction Model}
\label{subsec:direct-attraction}
In this subsection we consider a variant of the model, where a player payoff is the proportion of users for whom she provides the \textit{most accurate} prediction. Notice that this payoff scheme is aligned with the one considered by \citet{immorlica2011dueling}. We show a negative result for this payoff scheme -- a PNE may not exist without further assumptions on $\mD$. Surprisingly, this is true even when players are restricted to linear strategies. Finally, we highlight a special case where a PNE does exist. To facilitate understanding, we denote the elements we reconsider in this subsection by a tilde.

First, we formalized the payoff scheme considered in this subsection. Let $\tmZ = \mX \times \mY$, $\tmD$ be the marginal distribution of $\mD$ over $\tmZ$ and let $\tmS$ denote a sample from $\tmZ$. Given $\tz = (x,y)\in \tmZ$ and  a strategy profile $\bl h$, Denote by $B(\tz ; \bl h)$ the number of players who provide the most accurate prediction. In other words,
\begin{equation}
\label{eq:direct-B}
B(\tz ; \bl h) = \left\{ i\in \mN : \abs{h_i(x)-y} \leq \min_{i' \in \mN} \abs{h_{i'}(x)-y} \right\}.
\end{equation}
In addition, let
\[
\tw_i (\tz ; \bl h) = 
\begin{cases}
\frac{1}{\abs{B(\tz ; \bl h)}} & i\in B(\tz ; \bl h) \\
0 & \text{otherwise}
\end{cases}.
\]
Finally, the payoffs and the empirical payoffs are defined by
\begin{equation}
\label{eq-new payoffs}
\tpi_i (\bl h)= \E_{\tz\sim \tmD}\left[ \tw_i(\tz;\bl h) \right], 
\tpi_i^\mS (\bl h)= \E_{z \in \tmS}\left[ \tw_i(\tz;\bl h) \right].
\end{equation}
Indeed, this payoff depicts the case where each user is attracted to the player (or some, if there is more than one) that offered her the most accurate prediction. The induced game is a constant sum game, as now every user (or a point in $\tmZ$) grants a monetary unit to one of the players w.p. 1. Moreover, notice that the third entry of $\mZ$, the tolerance, is completely disregarded in this payoff scheme.

Recall that in Section \ref{sec:meta} we have demonstrated that players can use a sample drawn i.i.d. to learn an approximate PNE with high probability. However, as we show next, under the revised payoff scheme a PNE may not exist.
\begin{proposition}
\label{prop:no-eq-direct}
A game $\langle \tmZ,\tmD ,\mN,\mH,\tpi \rangle$ may not possess an approximate PNE for a constant approximation factor.
\end{proposition}
Proposition \ref{prop:no-eq-direct} implies that under the payoff scheme defined in Equation (\ref{eq-new payoffs}), players may not learn an $\epsilon$-PNE from a sample, regardless of the sample size. We prove Proposition \ref{prop:no-eq-direct} by providing a special case, which is given in Example \ref{example:no-eq-direct}.
\begin{example}
\label{example:no-eq-direct}
Let $\tmZ \subset \R^2$, and let $U\subset \tmZ$ be a union of squares as follows:
\[
U \defeq  \left( [0,1]\times [1,2]\right) \cup\left( [1,2]\times [0,1] \right) \cup \left( [2,3] \times [1,2] \right).
\]
See Figure \ref{fig:non-existence} for illustration. Let $\tmD$ be the uniform distribution over $U$, let $N=2$ (i.e., a two-player game) and let $\mH_1 = \mH_2$ be the linear strategy space. Observe that
\begin{claim}
\label{claim:no-eq-direct}
There exists no $\left(\frac{1}{6}-\epsilon\right)$-PNE in $\langle \tmZ,\tmD ,\mN,\mH,\tpi \rangle$, for any arbitrarily small constant $\epsilon>0$.
\end{claim}
The above Claim \ref{claim:no-eq-direct} proves Proposition \ref{prop:no-eq-direct}.
\end{example}
\begin{figure}[t]
\centering
\includegraphics[scale=1.]{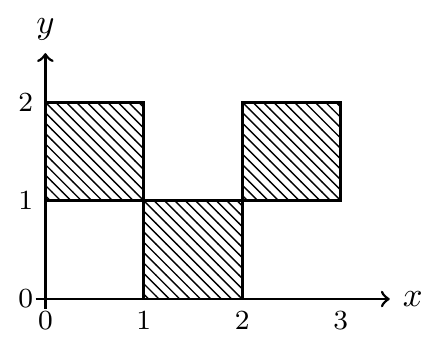}
\caption{The distribution $\tmD$ is the uniform distribution on the filled squares. As suggested by Claim \ref{claim:no-eq-direct},
for any selection of two linear functions, one of the players will have a beneficial deviation granting her at least $\frac{2}{3}-\epsilon$, for any arbitrarily small $\epsilon>0$. \label{fig:non-existence}}
\end{figure}

In fact, it is not immediate that the game proposed in Example \ref{example:no-eq-direct} even possesses a mixed Nash equilibrium, despite being a symmetric zero-sum two-player game. Observe that not only the strategy space of the players is infinite, but also player payoffs are discontinuous, so Glicksberg's theorem \cite{glicksberg1952further} cannot be applied. Since this work focuses on pure strategy profiles, we leave the full analysis of this extension for future work.

Before the end of this section, we highlight a class of games that do possess a PNE.
The setting depicted in Example \ref{example:no-eq-direct} is carefully crafted to demonstrate the negative effects of $\tpi$ as defined in Equation (\ref{eq-new payoffs}). Arguably, linear strategies are bad predictors for the structure of the data in this example, which is highly non-linear. One may ask whether a PNE must exist in the \textit{realizable} case.
\begin{definition}
A game $\langle \tmZ,\tmD ,\mN,\mH,\tpi \rangle$ is \textit{realizable} if for every $i\in \mN$ and  $(x,y)\in\tmZ$ there exists $h_i^*$ such that $h_i^*(x)=y$.
\end{definition}

In a realizable game, every player has a strategy that predicts the label of every instance correctly. Clearly, $(h_1^*,\dots,h_N^*)$ is a PNE of the game, and so at least one PNE exists. We now show that an approximate PNE of a realizable game can be found using the technique we developed in Section \ref{sec:meta}. Given $\epsilon,\delta$, let $m=m_{\frac{\epsilon}{2},\delta}$ that satisfies Equation (\ref{eq:lemma:deltaandm}), and let $\tmS$ be a sample of size $m$ from $\tmD$. Due to the realizability assumption, every player $i$ has a strategy that accurately predicts the value of every instance in the sample; let $h_i$ be one such strategy (this does not imply that $h_i=h_i^*$). Importantly, $h_i$ can be found using the same best response oracle employed earlier in Algorithm \ref{algorithm:betterres}, when augmenting every point in $\tmS$ to have 0 in its third entry (i.e., zero tolerance). 
The above discussion is formalized via Algorithm \ref{algorithm:betterres-direct}. We show that it outputs an $\epsilon$-PNE with probability of at least $1-\delta$.
\begin{lemma}
\label{lemma:realizable}
Let $\langle \tmZ,\tmD^m ,\mN,\mH,\tpi \rangle$ be a realizable game, $\epsilon,\delta\in (0,1)$, and let $\bl h$ be the output of Algorithm \ref{algorithm:betterres-direct}. With probability of at least $1-\delta$, $\bl h$ is an $\epsilon$-PNE.
\end{lemma}
\begin{algorithm}[t]
\DontPrintSemicolon
\caption{Approximate PNE w.h.p. for the realizable case\label{algorithm:betterres-direct}}
\KwIn{$\delta, \epsilon \in (0,1)$
}
\KwOut{a strategy profile $\bl h$}
set $m = m_{\frac{\epsilon}{2},\delta}$  \tcp*{the minimal integer $m$ satisfying Equation (\ref{eq:lemma:deltaandm})}
sample $\tmS$ from $\tmD^m$\;
for every $i\in \mN$, find $h_i$ such that $\sum_{(x,y)\in \tmS}\ind_{h_i(x)\neq y}=0$ \;
\Return{$\bl h=\left(h_1,\dots,h_N\right)$}
\end{algorithm}
\section{Discussion}
\label{sec:discussion}
We have presented a novel setting that deepens the intersection between learning theory and game theory. We used the seminal framework of 
\citeauthor{valiant1984theory} to model an environment where users seek predictions and players provide them, where players' payoffs are defined in the spirit of Dueling Algorithms \cite{immorlica2011dueling}. We proved the existence of a PNE, bounded the number of steps until convergence and formalized an algorithm that computes a PNE assuming better-response oracles. We provided an oracle for linear strategy space, and demonstrated the equilibrium structure for linear two-player games in several settings. We have also considered cases where players cannot learn a PNE, despite that user distribution is fully known, and provided preliminary results on an interesting variant of the model, which will appear in future work.

As mentioned in Section \ref{subsec:model}, our analysis assumes players have better-response oracles. In fact, our model and results are valid for a much more general scenario, as described next. Consider the case where players only have heuristics for finding a better response. After running heuristic better-response dynamics and obtaining a strategy profile, the uniform convergence property implies that the payoffs with respect to the whole population are guaranteed to be close to their empirical counterparts, w.h.p.; therefore, our analysis is still meaningful even if players cannot maximize their empirical payoff efficiently. The bounds on the required sample size we obtained in Section \ref{sec:meta} and the rate of convergence are relevant for this case as well. 

Ultimately, we remark that regression equilibrium can be employed to develop piece-wise linear models, or other mixture models in general, as illustrated in Subsection \ref{SUBSEC:SIM} and Figure \ref{fig:two-dim}. This is a future direction of independent interest.

\section*{Acknowledgments}\label{sec:Acknowledgments}
We thank the anonymous reviewers for their helpful comments. Example \ref{example:no-eq-direct} is a simplified version of the one we had originally, and is due to Yakov Babichenko. The work of O. Ben-Porat was supported by a JPMorgan Chase \& Co. PhD Fellowship. This project has received funding from the European Research Council (ERC) under the European Union's Horizon 2020 research and innovation programme (grant agreement n$\degree$  740435).

%\bibliography{bibdb}
\input{bibdb.bbl}

{\ifnum\Includeappendix=1{ %statrting appendices

\appendix
\section{Proof of Lemma \ref{LEMMA:REGTOCLASS}}

\begin{proof}[\textbf{Proof of Lemma \ref{LEMMA:REGTOCLASS}}]
\omer{read it all and verify}
First, we define two auxiliary classes of binary functions $\mG^{\geq},\mG^{\leq}$ such that
\begin{align}
\label{eq:ggeqeqdef}
&\mG^{\geq} =\{g_h^\geq (x,r)= \ind_{ h(x) \geq r} \mid h\in \mH_i ,(x,r)\in \mX\times \mathbb{R}  \},\nonumber \\
&\mG^{\leq} =\{g_h^\leq (x,r)=\ind_{ h(x) \leq r} \mid h\in \mH_i ,(x,r)\in \mX\times \mathbb{R}  \}.
\end{align}
\begin{claim}
\label{claim:auxiliaryg}
$\vc(\mG^{\geq})=\vc(\mG^{\leq}) = d_i$. 
\end{claim}
The proof of Claim \ref{claim:auxiliaryg} appears after the proof of this lemma. Next, we wish to bound the growth function of $\mF_i$ using the growth function of $\mG^{\geq}$ and $\mG^{\leq}$.
\begin{claim}
\label{claim:fgrowth}
$\Pi_{\mF_i}(m) \leq \Pi_{\mG^{\geq}}(m) \cdot \Pi_{\mG^{\leq}}(m)$.
\end{claim}
The proof of Claim \ref{claim:fgrowth} appears after the proof of this lemma. We are now ready for the final argument. By the Sauer-Shelah lemma we know that every $m$ satisfying $2^m > \Pi_{\mF_i}(m)$ is an upper bound on $\vc(\mF_i)$ \cite{sauer1972density}. In particular, for $m=10d_i$ we have
\begin{align}
2^{5d_i}&=(31+1)^{d_i}=\sum_{j=0}^{d_i}31^j 1^{d_i-j} {d_i \choose j} \stackrel {\text{Claim \ref{claim:binomialfactors}}}{\geq }
\sum_{j=0}^{d_i} \left( \frac{31 d_i}{j}\right)^j  \\
&> \sum_{j=0}^{d_i} \left( \frac{10 e d_i}{j}\right)^j \stackrel {\text{Claim \ref{claim:binomialfactors}}}{\geq } \sum_{j=0}^{d_i} {10d_i \choose j} \nonumber \\ 
&\geq \Pi_{\mG^{\geq}}(10d_i)=\Pi_{\mG^{\leq}}(10d_i);
\end{align}
therefore,
\[
\Pi_{\mF_i}(10d_i) \leq \Pi_{\mG^{\geq}}(10d_i)\Pi_{\mG^{\leq}}(10d_i) < 2^{10d_i} .
\]

\end{proof}

\begin{proof}[\textbf{Proof of Claim \ref{claim:auxiliaryg}}]

We prove the claim for $\mG^{\geq}$, and by symmetric arguments one can show it holds for $\mG^{\leq }$ as well.

Since $\pdim(\mH_i)=d_i$, for every $m\leq d_i$ there is a sample $\mS=(x_1,\dots  x_m)\in \mX^m$ and a witness $\bl r=(r_1,\dots r_m) \in \R^m$ such that for every binary vector $\bl b\in \{-1,1\}^m$ there is a function $h_{\bl b}$ for which $\sign(h_{\bl b}(x_j)-r_j)=b_j$ for all $j\in [m]$. Denote 
\[
\mS' = \left( (x_1,r_1),\dots,(x_m,r_m) \right),
\]
and focus on a particular $\bl b\in \{-1,1\}^m$. For every $j\in [m]$ such that $b_j=1$ we have
\[
\sign(h_{\bl b}(x_j)-r_j)=1 \Rightarrow  h_{\bl b}(x_j)-r_j >0 \Rightarrow\ind_{h_{\bl b}(x)\geq r}(x_j,r_j)=1.
\]
In addition, if $b_j=-1$ then
\[
\sign(h_{\bl b}(x_j)-r_j)=-1 \Rightarrow  h_{\bl b}(x_j)-r_j <0 \Rightarrow\ind_{h_{\bl b}(x)\geq r}(x_j,r_j)=0.
\]
This is true for every $\bl b$; therefore, we showed that $\mG^{\geq}$ shatters $\mS'$.

In the opposite direction, assume by contradiction that $\mG^{\geq}$ shatters $\mS'=\left( (x_1,r_1),\dots,(x_m,r_m) \right)$ for $m\geq d_i+1$. Let $H=\{ h_{\bl b} \}_{\bl b\in \{-1,1\}} \subset \mH_i$ be a set of functions such that for every $\bl b$ there exists exactly one function $h_{\bl b}\in H$ satisfying $\ind_{h_{\bl b}(x)\geq r}(x_j,r_j)=1$ if $b_j=1$ and $\ind_{h_{\bl b}(x)\geq r}(x_j,r_j)=0$ if $b_j=0$. Notice that by definition of the VC dimension, such an $H$ must exist, and that $\abs{H}=2^m$.

One cannot claim directly that $\mH_i$ pseudo-shatters $\mS=(x_1,\dots,x_m)$ with witness $\bl r=(r_1,\dots r_m)$, since $h_{\bl b}(x_j)=r_j$ may hold for $b_j=1$, but we need $h_{\bl b}(x_j)$ to be strictly greater than $r_j$; therefore, we construct a new witness: let $a_j$ be such that
\[
a_j= \max_{h_{\bl b}\in H, b_j=-1} h_{\bl b}(x_j).
\]
Notice that $\abs{H}$ is finite so the maximum is well defined. In addition, $a_j < r_j$ since $h_{\bl b}(x_j) <r_j$ for every $\bl b$ such that $b_j=-1$ (recall that if $b_j=0$ then  $\ind_{h_{\bl b}(x)\geq r}(x_j,r_j)=0$ ).

Denote $\bl r^* = \frac{\bl a +\bl r}{2}$. Next, we claim that $\mH_i$ pseudo-shatters $\mS$ with the witness $\bl r^*$. Fix $\bl b\in\{-1,1\}^m$. If $b_j=-1$,
\[
\ind_{h_{\bl b}(x)\geq r}(x_j,r_j)=0 \Rightarrow h_{\bl b}(x_j) \leq  a_j \Rightarrow h_{\bl b}(x_j) < r_j^* \Rightarrow \ind_{h_{\bl b}(x)\geq r}(x_j,r^*_j)=0.
\]

On the other hand, if $b_j=1$, we have
\[
\ind_{h_{\bl b}(x)\geq r}(x_j,r_j)=1 \Rightarrow h_{\bl b}(x_j) \geq  r_j \Rightarrow h_{\bl b}(x_j) > r_j^* \Rightarrow \ind_{h_{\bl b}(x)\geq r}(x_j,r^*_j)=1.
\]
Combining these two equations, we get that $\sign(h_{\bl b}(x_j)-r_j^*)=b_j$ for all $j\in [m]$. Consequently, $\mH_i$ pseudo-shatters $\mS$ with witness $\bl r^*$; hence we obtained a contradiction.

Overall, we showed that $\vc(\mG^{\geq}) \geq d_i$ and $\vc(\mG^{\geq}) \leq d_i$; hence $\vc(\mG^{\geq})=d_i$.
\end{proof}

\begin{proof}[\textbf{Proof of Claim \ref{claim:fgrowth}}]
Denote by $\mS=(x_j,y_j,t_j)_{j=1}^m \in \mZ^m$ an arbitrary sample, and let $\mF_i \cap \mS$ be the restriction of $\mF_i$ to $\mS$. Formally,
\[
\mF_i \cap \mS = \left\{(f_h(z_1),\dots,f_h(z_m))\mid f_h \in \mF_i  \right\}.
\] 

In addition, denote by $G^\geq $ the restriction of $\mG^\geq$ to $(x_j,y_j-t_j)_{j=1}^m$, and similarly let $G^\leq $ be the restriction of $\mG^\leq$ to $(x_j,y_j+t_j)_{j=1}^m$. We now show a one-to-one mapping $M: {\mF_i} \cap {\mS} \rightarrow G^\geq \times G^\leq$, implying that 
\begin{equation}
\label{eq:mnkasd}
\abs{{\mF_i} \cap {\mS} } \leq \abs{G^\geq \times G^\leq }
\end{equation}
holds, thereby proving the assertion. Notice that for every $f_h \in \mF_i$ such that $f_h(z)=1$ we have $\mI(z,h)=1$ for the corresponding $h\in \mH_i$; thus 
\begin{equation}
\label{eq:uioet}
-t_j\leq h(x_j)-y_j \leq t_j  
%\Rightarrow\begin{cases} h(x_j) \leq y_j+t_j  \\ h(x_j) \geq y_j-t_j  \\ \end{cases}
\Rightarrow
\begin{cases}
\ind_{ h(x_j) \leq y_j+t_j} = 1 \\
\ind_{ h(x_j) \geq y_j-t_j} = 1
\end{cases}
\Rightarrow
\begin{cases}
g_h^\leq (x_j,y_j+t_j)=1 \\
g_h^\geq (x_j,y_j-t_j)=1
\end{cases}. \footnote{Recall the definition of $g_h^\leq,g_h^\geq$ in Equation (\ref{eq:ggeqeqdef}).}
\end{equation}
Alternatively, if $f_h(z_j)=0$, we have $\mI(z_j,h)=0$ and
\begin{align}
\label{eq:uigfgfg}
&(h(x_j)-y_j < -t_j)  \lor ( h(x_j)-y_j > t_j) \Rightarrow 
(\ind_{ h(x_j) < y_j-t_j} = 1) \lor  (\ind_{ h(x_j)> y_j+t_j} =1) \nonumber \\
& \Rightarrow (\ind_{ h(x_j) \geq y_j-t_j} = 0) \lor  (\ind_{ h(x_j) \leq y_j+t_j} = 0) \Rightarrow (g_h^\geq (x_j,y_j-t_j)=0) \lor  (g_h^\leq (x_j,y_j+t_j) = 0).
\end{align}
By Equations (\ref{eq:uioet}) and (\ref{eq:uigfgfg}) we have
\begin{equation}
\label{eq:adfjnaf}
\begin{cases}
f_h(z_j)=1 \Rightarrow (g_h^\leq (x_j, y_j+t_j),g_h^\geq(x_j, y_j-t_j))=(1,1)\\
f_h(z_j)=0 \Rightarrow (g_h^\leq(x_j, y_j+t_j),g_h^\geq (x_j, y_j-t_j))\in \{(0,0),(0,1),(1,0)\}
\end{cases}.
\end{equation}
We define the mapping $M$ such that every vector $\left(\mI(z_1,h_1),\dots,\mI(z_m,h_1) \right)\in \mF_i \cap \mS$ is mapped to 
\[
\left(g_h^\geq (x_1, y_1-t_1),\dots, g_h^\geq (x_m, y_m-t_m),g_h^\leq (x_1, y_1+t_1),\dots ,g_h^\leq (x_m, y_m+t_m) \right)
\in G^\geq \times G^\leq.
\]
Namely, every vector obtained by applying $f_h$ on the sample $\mS$ is mapped to the vector formed by concatenating the two corresponding (same $h$) vectors from $G^\leq $ and $G^\geq $. Let $\bl b^1,\bl b^2 \in {\mF_i} \cap \mS$ such that  $b^1_j \neq b^2_j$ for at least one index $j\in[m]$, and w.l.o.g. let $b^1_j=1$. Since $b^1_j=f_h(z_j)=\mI(z_j,h)$, Equation (\ref{eq:adfjnaf}) implies that $M(\bl b^1)_j = M(\bl b^1)_{j+m} = 1$, while at least one of $\{  M(\bl b^2)_j, M(\bl b^2)_{j+m}\}$ equals zero; thus $M(\bl b^1) \neq M(\bl b^2)$. Hence $M$ is an injection. 

Ultimately, notice that $\mS$ is arbitrary; thus
\[\Pi_{\mF_i}(m) =\max_{\mS\in \mZ^m} \abs{{\mF_i}\cap \mS }\leq \abs{G^\geq \times G^\leq } = \abs{G^\geq } \cdot \abs{ G^\leq } \leq \Pi_{\mG^{\geq}}(m) \cdot \Pi_{\mG^{\leq}}(m).\]

\end{proof}

\begin{claim}
\label{claim:binomialfactors}
$\left(\frac{n}{k}  \right)^k \leq {n \choose k} \leq \left(\frac{e n}{k}  \right)^k   $. 
\end{claim}
\begin{proof}[\textbf{Proof of Claim \ref{claim:binomialfactors}}] We prove the two claims separately.

$\bullet$ $\left(\frac{n}{k}  \right)^k \leq {n \choose k}$: fix $n$. We prove by induction for $k\leq n$. The assertion holds for $k=1$. For $k\geq 2$ and every $m$ such that $0<m<k\leq $ we have
\begin{equation*}
\label{eq:gkkmpdf}
k \leq n \Rightarrow \frac{m}{n} \leq \frac{m}{k} \Rightarrow 1 - \frac{m}{k} \leq 1- \frac{m}{n} \Rightarrow \frac{k-m}{k} \leq \frac{n-m}{n} \Rightarrow \frac{n}{k} \leq \frac{n-m}{k-m};
\end{equation*}
thus
\[
\left(\frac{n}{k}\right)^k = \frac{n}{k} \cdots \frac{n}{k}\leq \frac{n}{k}\frac{n-1}{k-1}\cdots\frac{n-k+1}{k-k+1} = \binom{n}{k}.
\]

$\bullet$ $ {n \choose k} \leq \left(\frac{e n}{k}  \right)^k  $:  since $e^k=\sum_{i=0}^\infty \frac{k^i}{i !}$ (the Taylor expansion of $e^k$), we have $e^k > \frac{k^k}{k!}$; thus, $\frac{1}{k!} < \left(\frac{e}{k}\right)^k$. As a result,
\[
\binom{n}{k}=\frac{n\cdot (n-1) \cdots (n-k+1)}{k!} \leq \frac{n^k}{k!}<\left(\frac{en}{k}   \right)^k.
\]
\end{proof}

\section{Omitted proofs}

\begin{proof}[\textbf{Proof of Lemma \ref{lemma:growthsumvc}}]
Recall that the Sauer-Shelah lemma implies that $\Pi_{\mF_i}(m) \leq (em/d_i)^{d_i}$ for $m>d_i+1$. Since $\abs{\mF \cap \mS}=\prod_{i=1}^N \abs{\mF_i \cap \mS}$, we have
\[
\Pi_{\mF}(m) \leq \prod_{i=1}^N \Pi_{\mF_i}(m) \leq \prod_{i=1}^N (em/d_i)^{10d_i}\leq  \prod_{i=1}^N (em)^{10d_i}=(em)^{10\sum_{i=1}^N d_i}.
\]
\end{proof}

\begin{proof}[\textbf{Proof of Claim \ref{claim:size and f}}]
To prove the claim, we show a one-to-one function from $\mW \cap \mS$ to $\mF\cap \mS$ for any arbitrary $\mS$. In particular, it is sufficient to show a mapping $ w(z;\bl h) \mapsto \mI(z,\bl h)$. Let $M$ be a mapping from $\{1,\frac{1}{2},\dots,\frac{1}{N},0\}^N$ to $\{0,1\}^N$ defined by
\[
M(\bl v) = \begin{cases}
\bl 0 & \text{if } \norm{\bl v}_1=0\\
\norm{\bl v}_0 \bl v  &\text{otherwise}
\end{cases}.
\]
For instance, if $N=4$ and $\bl v = (\frac{1}{2},\frac{1}{2},0,0)$, then $M(\bl v) = (1,1,0,0)$. Clearly, $M$ is a one-to-one mapping from $\text{Image}(w(z;\bl h))$ to $\text{Image}(\mI(z,\bl h))$. For every vector $\bl w\in \mW \cap \mS$, there exists $\bl h_{\bl w} \in \mH$ such that
\[
\bl w = (w(z_1;\bl h_{\bl w}),\dots,w(z_m;\bl h_{\bl w})).
\]
By applying $M$ on every entry of $\bl w$, we obtain $\bl v$ such that
\[
\bl v = (\mI(z_1;\bl h_{\bl w}),\dots,\mI(z_m;\bl h_{\bl w})).
\]
Ultimately, since $\bl h_{\bl w}\in \mH$, we conclude that $\bl v \in  \mF \cap \mS$.

In fact, by considering $M^{-1}$, $M^{-1}:\{0,1\}^N \rightarrow \{1,\frac{1}{2},\dots,\frac{1}{N},0\}^N$ such that for every $\bl u \in \{0,1\}^N$,
\[
M^{-1}(\bl u) = 
\begin{cases}
\bl 0 & \text{if } \norm{\bl u}=0\\
\frac{\bl u }{\norm{\bl u}_1} &\text{otherwise}
\end{cases},
\]
one can show that $\Pi_{\mF}(m) \leq \Pi_{\mW}(m)$ holds as well; therefore, $\Pi_{\mW}(m) = \Pi_{\mF}(m)$.
\end{proof}

\begin{proof}[\textbf{Proof of Lemma \ref{lemma:uniconvergenceoneplayer}}]
\newcommand{\hbad}{\tilde{ \bl h}(\mS)}
The proof follows closely the four steps in the proof of the classical uniform convergence theorem for binary functions (see, e.g., \cite{anthony2009neural,vapnik1971uniform}). The only steps that need modification are Steps 3 and 4, but we present the full proof for completeness.

\textit{Step 1 -- Symmetrization:}
First, we want to show that
\begin{equation}
\Pr_{\mS\sim \mD^m}\left(\exists \bl{h} : \abs{\pi_i(\bl h)-\pi_i^{\mS}(\bl h)} \geq  \epsilon \right) \leq 
2 \Pr_{(\mS,\mS')\sim \mD^m}\left(\exists \bl{h} : \abs{\pi_i^{\mS}(\bl h)-\pi_i^{\mS'}(\bl h)} \geq  \frac{\epsilon}{2} \right) \label{eq:steponeaa}.
\end{equation}

 For each $\mS$, let $\hbad$ be a function for which $\abs{\pi_i(\hbad)-\pi_i^{\mS}(\hbad)} \geq \epsilon$ if such a function exists, and any other fixed function in $\mH$ otherwise. Notice that if $\abs{\pi_i(\hbad)-\pi_i^{\mS}(\hbad)} \geq \epsilon$ and  $\abs{\pi_i(\hbad)-\pi_i^{\mS'}(\hbad)} \leq  \frac{\epsilon}{2}$, then $\abs{\pi_i^{\mS}(\hbad)-\pi_i^{\mS'}(\hbad)} \geq   \frac{\epsilon}{2}$ (triangle inequality); thus,
\begin{align*}
&\Pr_{(\mS,\mS')\sim \mD^m}\left(\exists \bl{h} : \abs{\pi_i^{\mS}(\bl h)-\pi_i^{\mS'}(\bl h)} \geq  \frac{\epsilon}{2} \right) \\
&\geq \Pr_{(\mS,\mS')\sim \mD^m}\left(  \abs{\pi_i^{\mS}(\hbad)-\pi_i^{\mS'}(\hbad)} \geq  \frac{\epsilon}{2} \right) \\
&\geq \Pr_{(\mS,\mS')\sim \mD^m}\left( \abs{\pi_i(\hbad)-\pi_i^{\mS}(\hbad)} \geq \epsilon \cap \abs{\pi_i(\hbad)-\pi_i^{\mS'}(\hbad)}  \leq \frac{\epsilon}{2} \right) \\
&=\E_{\mS\sim \mD^m} \left[ \ind\left( \abs{\pi_i(\hbad)-\pi_i^{\mS}(\hbad)} \geq \epsilon  \right)\Pr_{\mS'\mid \mS}\left( \abs{\pi_i(\hbad)-\pi_i^{\mS'}(\hbad)}  \leq \frac{\epsilon}{2}    \right) \right] \\ 
&\geq \frac{1}{2}\Pr_{\mS \sim \mD^m}\left(  \abs{\pi_i(\hbad)-\pi_i^{\mS}(\hbad)} \geq \epsilon \right)\\
&= \frac{1}{2}\Pr_{\mS \sim \mD^m}\left(  \exists \bl{h} : \abs{\pi_i(\bl h)-\pi_i^{\mS}(\bl h)} \geq \epsilon \right),
\end{align*}

since $\mS,\mS'$ are independent and due to Claim \ref{claim:rqep}.

\textit{Step 2 -- Permutations:} We denote $\Gamma_{2m}$ as the set of all permutations of $[2m]$ that swap $i$ and $m+i$ in some subset of $[m]$. Namely,
\[
\Gamma_{2m} = \{\sigma\in  \Pi([2m]) \mid \forall i\in[m] :\sigma(i) = i \lor \sigma(i)= m+i; \forall i,j \in [2m]: \sigma(i) = j \Leftrightarrow \sigma(j)=i   \},
\]
where $\Pi([2m])$ denotes the set of permutations over $[2m]$. In addition, for $\mS = (z_1,\dots,z_{2m})$, let $\sigma(\mS)=(z_{\sigma(1)},\dots,z_{\sigma(2m)})$. Notice that for every $\sigma \in \Gamma_{2m}$ it holds that
\[
\Pr_{(\mS,\mS')\sim \mD^m}\left(\exists \bl{h} : \abs{\pi_i^{\mS}(\bl h)-\pi_i^{\mS'}(\bl h)} \geq  \frac{\epsilon}{2} \right) = \Pr_{(\mS,\mS')\sim \mD^m}\left(\exists \bl{h} : \abs{\pi_i^{\sigma(\mS)}(\bl h)-\pi_i^{\sigma(\mS')}(\bl h)} \geq  \frac{\epsilon}{2} \right);
\]
hence,
\begin{align}
&\Pr_{(\mS,\mS')\sim \mD^{2m}}\left(\exists \bl{h} : \abs{\pi_i^{\mS}(\bl h)-\pi_i^{\mS'}(\bl h)} \geq  \frac{\epsilon}{2} \right) \nonumber\\
& =\frac{1}{2^m}\sum_{\sigma\in \Gamma_{2m}} \Pr_{(\mS,\mS')\sim \mD^{2m}}\left(\exists \bl{h} : \abs{\pi_i^{\sigma(\mS)}(\bl h)-\pi_i^{\sigma(\mS')}(\bl h)} \geq  \frac{\epsilon}{2} \right)\nonumber\\
& =\frac{1}{2^m}\sum_{\sigma\in \Gamma_{2m}} \E_{(\mS,\mS')\sim \mD^{2m}}\left[\ind_{\exists \bl{h} : \abs{\pi_i^{\sigma(\mS)}(\bl h)-\pi_i^{\sigma(\mS')}(\bl h)} \geq  \frac{\epsilon}{2}} \right]\nonumber\\
& =\E_{(\mS,\mS')\sim \mD^{2m}}\left[\frac{1}{2^m}\sum_{\sigma\in \Gamma_{2m}} \ind_{\exists \bl{h} : \abs{\pi_i^{\sigma(\mS)}(\bl h)-\pi_i^{\sigma(\mS')}(\bl h)} \geq  \frac{\epsilon}{2}} \right]\nonumber\\
& =\E_{(\mS,\mS')\sim \mD^{2m}}\left[\Pr_{\sigma\in \Gamma_{2m}} \left(\exists \bl{h} : \abs{\pi_i^{\sigma(\mS)}(\bl h)-\pi_i^{\sigma(\mS')}(\bl h)} \geq  \frac{\epsilon}{2}\right) \right]\nonumber\\
&\leq \sup_{(\mS,\mS')\sim \mD^{2m}}\left[\Pr_{\sigma\in \Gamma_{2m}} \left(\exists \bl{h} : \abs{\pi_i^{\sigma(\mS)}(\bl h)-\pi_i^{\sigma(\mS')}(\bl h)} \geq  \frac{\epsilon}{2}\right) \right] \label{eq:steptwoperm}.
\end{align}

\textit{Step 3 -- Reduction to finite class:} Fix $(\mS,\mS')$ and consider a random draw of $\sigma \in \Gamma_{2m}$. For each strategy profile $\bl h$, the quantity $\abs{\pi_i^{\sigma(\mS)}(\bl h)-\pi_i^{\sigma(\mS')}}$ is a random variable. Since $\mW\cap \mS$ represents the number of distinct strategy profiles in the empirical game over $\mS$ (see Subsection \ref{subsec:sample}), there are at most $\Pi_{\mW}(2m)$ such random variables.

\begin{align}
&\Pr_{\sigma\in \Gamma_{2m}} \left(\exists \bl{h} : \abs{\pi_i^{\sigma(\mS)}(\bl h)-\pi_i^{\sigma(\mS')}(\bl h)} \geq  \frac{\epsilon}{2}\right)  \nonumber\\
&\leq \Pi_{\mW}(2m) \sup_{\bl h \in \mH} \Pr_{\sigma\in \Gamma_{2m}} \left(\abs{\pi_i^{\sigma(\mS)}(\bl h)-\pi_i^{\sigma(\mS')}(\bl h)} \geq  \frac{\epsilon}{2}\right) \label{hehehehe}
\end{align}

\textit{Step 4 -- Hoeffding's inequality:} By viewing the previous equation as Rademacher random variables, we have
\begin{align*}
&\Pr_{\sigma\in \Gamma_{2m}} \left(\abs{\pi_i^{\sigma(\mS)}(\bl h)-\pi_i^{\sigma(\mS')}(\bl h)} \geq  \frac{\epsilon}{2}\right) \\
&= \Pr_{\sigma\in \Gamma_{2m}} \left(\frac{1}{m}\abs{\sum_{j=1}^m w_i(z_{\sigma(j)},\bl h)-w_i(z_{\sigma(j+m)},\bl h)} \geq  \frac{\epsilon}{2}\right)\\
& =\Pr_{r\in\{-1,1\}^m}\left(\frac{1}{m}\abs{\sum_{j=1}^m r_j\left(w_i(z_j,\bl h)-w_i(z_{j+m},\bl h)\right)} \geq  \frac{\epsilon}{2}\right).
\end{align*}
Observe that for every $j$ it holds that $r_j\left(w_i(z_j,\bl h)-w_i(z_{j+m},\bl h)\right) \in [-1,1]$, and
\[
\E_{r_j\in\{-1,1\}}\left[ r_j\left(w_i(z_j,\bl h)-w_i(z_{j+m},\bl h)\right) \right]=0
\]
holds due to symmetry. By applying Hoeffding's inequality we obtain
\begin{align}
\label{eq:sfourfin}
 \Pr_{r\in\{-1,1\}^m}\left(\frac{1}{m}\abs{\sum_{j=1}^m r_j\left(w_i(z_j,\bl h)-w_i(z_{j+m},\bl h)\right)} \geq  \frac{\epsilon}{2}\right) \leq 2e^{-\frac{m\epsilon^2}{8}}.
\end{align}
Finally, by combining Equations (\ref{eq:steponeaa}),(\ref{eq:steptwoperm}),(\ref{hehehehe}) and (\ref{eq:sfourfin}) we derive the desired result.
\end{proof}

\begin{proof}[\textbf{Proof of Theorem \ref{thm:unionbound}}]
The theorem follows immediately by applying the union bound on the inequality obtained in Lemma \ref{lemma:uniconvergenceoneplayer} and by substituting $\Pi_\mW(2m)$ according to Claim \ref{claim:size and f} and Lemma \ref{lemma:growthsumvc}.
\end{proof}

\begin{proof}[\textbf{Proof of Lemma \ref{lemma:PNEexistence}}]
The lemma is proven by showing that an induced game has a potential function $\Phi : \mH\cap \mS \rightarrow \mathbb R$, and relying on the results of \citet{monderer1996potential}. Namely, we show a function $\Phi$ such that for every $i,\bl h$ and $h_i'$ it holds that \[ \pi_i(\bl h)-\pi_i(h_i', \bl h_{-i}) = \Phi(\bl h)-\Phi(h_i', \bl h_{-i}). \]
Let $N(z_j;\bl h)$ denote the number of players satisfying the point $z_j$ under $\bl h$, namely $N(z_j;\bl h)=\sum_{i=1}^N \mI(z_j,h_{i})$. Next, let $\Phi(\bl h) \defeq \frac{1}{m}\sum_{j=1}^m \sum_{k=1}^{N(z_j;\bl h)} \frac{1}{k}$, and observe that
\begin{align*}
& \pi_i(\bl h)-\pi_i(h_i', \bl h_{-i})=\frac{1}{m}\sum_{j=1}^m w_i(z_j;\bl h) -   \frac{1}{m}\sum_{j=1}^m w_i(z_j;h_i', \bl h_{-i}) \\
& =\frac{1}{m}\sum_{j=1}^m \frac{\mI_i(z;h_i)}{\abs{N(z;\bl h)}}  -   \frac{1}{m}\sum_{j=1}^m  \frac{\mI_i(z;h_i')}{\abs{N(z;h_i', \bl h_{-i})} }+\frac{1}{m}\sum_{j=1}^m \sum_{k=1}^{N(z_j;\bl h_{-i})} \frac{1}{k}-\frac{1}{m}\sum_{j=1}^m \sum_{k=1}^{N(z_j;\bl h_{-i})} \frac{1}{k} \\
& = \frac{1}{m}\sum_{j=1}^m \sum_{k=1}^{N(z_j;\bl h)} \frac{1}{k} - \frac{1}{m}\sum_{j=1}^m \sum_{k=1}^{N(z_j;h_i, \bl h_{-i})} \frac{1}{k}=\Phi(\bl h)-\Phi(h_i', \bl h_{-i}).
\end{align*}
Since the empirical game is finite (i.e., finite number of players with finite strategy space for each player), we conclude that it possesses at least one PNE.
\end{proof}

\begin{proof}[\textbf{Proof of Lemma \ref{lemma:betterconverge}}]
Let $N(z_j;\bl h)$ as in the proof of Lemma \ref{lemma:PNEexistence}. In each iteration of the dynamics it holds that
\begin{align}
\label{eq:bnfesw}
\Phi(h_i', \bl h_{-i}) - \Phi(\bl h) = \pi_i(h_i', \bl h_{-i}) - \pi_i(\bl h) \geq \epsilon.
\end{align}
Notice that 
\begin{equation}
\label{eq:aewrbg}
\Phi(\bl h) =  \frac{1}{m}\sum_{j=1}^m \sum_{k=1}^{N(z_j;\bl h)} \frac{1}{k} \leq  \frac{1}{m}\sum_{j=1}^m \sum_{k=1}^{N} \frac{1}{k} \leq \ln N +1.
\end{equation}
Since the potential is bounded by $\ln N +1$ and increases by at least $\epsilon$ per iteration throughout the dynamics, after at most $\frac{\ln N +1}{\epsilon}$ iterations, it will reach its maximum value, thereby obtaining an $\epsilon$-PNE.
\end{proof}

\begin{proof}[\textbf{Proof of Lemma \ref{lemma:deltaandm}}]
By Equation (\ref{eq:inthmunionbound}), we look for $m$ that satisfies
\[
4N (2em)^{10\sum_{i=1}^N d_i}e^{-\frac{\epsilon^2 m}{8}} \leq \delta
\]
for given $\epsilon,\delta$; thus
\begin{align}
&(2em)^{10\sum_{i=1}^N d_i}e^{-\frac{\epsilon^2 m}{8}} \leq \frac{\delta}{4N} \nonumber \\
&\Rightarrow 10d \log(2em)-\frac{\epsilon^2 m}{8} \leq \log \frac{\delta}{4N}\nonumber \\
& \frac{\epsilon^2 m}{8} \geq 10d\log(2em)-\log \frac{\delta}{4N}\nonumber \\
& m\geq \frac{80d\log(2em)}{\epsilon^2}-\frac{8}{\epsilon^2}\log \frac{\delta}{4N}\nonumber \\
& m\geq \frac{80d}{\epsilon^2}\log(m)+\frac{80d\log(2e)}{\epsilon^2}+\frac{8}{\epsilon^2}\log \frac{4N}{\delta}. \label{eq:dgfjnkindfjsnk}
\end{align}
Next, 
\begin{claim}[\cite{shalev2014understanding}, Section A]
\label{claim:fromsssappendix}
Let $a\geq 1$ and $b>0$. If $m \geq 4a\log(2a)+2b$, then $m\geq a \log m +b$.
\end{claim}
Set $a = \frac{80d}{\epsilon^2}$ and $b=\frac{80d\log(2e)}{\epsilon^2}+\frac{8}{\epsilon^2}\log \frac{4N}{\delta}$. Due to Claim \ref{claim:fromsssappendix}, we know that every $m$ that satisfies
\[
m \geq \frac{320d}{\epsilon^2} \log\left( \frac{160d}{\epsilon^2} \right) +\frac{160d\log(2e)}{\epsilon^2}+\frac{16}{\epsilon^2}\log\left( \frac{4N}{\delta} \right)
\]
also satisfies Equation (\ref{eq:dgfjnkindfjsnk}).
\end{proof}

\begin{proof}[\textbf{Proof of Lemma \ref{lemma:empiseq}}]
Notice that for every $i,h'_i$ it holds that
\begin{align*}
\pi_i(h_i', \bl h_{-i})-\pi_i(\bl h)&=\pi_i(h_i', \bl h_{-i})-\pi_i^\mS(h_i', \bl h_{-i})+\pi_i^\mS(h_i', \bl h_{-i})-\pi_i(\bl h) \\
& \stackrel{\substack{\bl h \text{ is an } \\ \frac{\epsilon}{2}\text{-empirical-PNE}} }{\leq} \pi_i(h_i', \bl h_{-i})-\pi_i^\mS(h_i', \bl h_{-i})+\pi_i^\mS(\bl h)-\pi_i(\bl h)+\frac{\epsilon}{2};
\end{align*}
therefore, if $\pi_i(h_i', \bl h_{-i})-\pi_i(\bl h)>\epsilon$ then at least one of $\pi_i(h_i', \bl h_{-i})-\pi_i^\mS(h_i', \bl h_{-i})> \frac{\epsilon}{4}$ or $\pi_i^\mS(\bl h)-\pi_i(\bl h) > \frac{\epsilon}{4}$ must hold. Overall,
\begin{align*}
&\Pr_{\mS\sim \mD^m}\left(\bl h \text{ is not an  } \epsilon \text{-PNE}  \right) = \Pr_{\mS\sim \mD^m}\left(\exists i \in \mN, h'_i \in \mH_i: \pi_i(h_i', \bl h_{-i})-\pi_i(\bl h) > \epsilon \right) \\
& \leq \Pr_{\mS\sim \mD^m}\left(\exists i \in \mN, h'_i \in \mH_i: \pi_i(h_i', \bl h_{-i})-\pi_i^\mS(h_i', \bl h_{-i})> \frac{\epsilon}{4} \text{ or } \pi_i^\mS(\bl h)-\pi_i(\bl h) > \frac{\epsilon}{4} \right) \\
& \leq \Pr_{\mS\sim \mD^m}\left(\exists i \in \mN, h'_i \in \mH_i: \abs{\pi_i(h_i', \bl h_{-i})-\pi_i^\mS(h_i', \bl h_{-i})> \frac{\epsilon}{4}} \text{ or } \abs{\pi_i^\mS(\bl h)-\pi_i(\bl h) > \frac{\epsilon}{4}} \right) \\
& \leq \Pr_{\mS\sim \mD^m}\left(\exists i\in\mN : \sup_{\bl h'' \in \mH}\abs{\pi_i(\bl h'')-\pi_i^{\mS}(\bl h'')} \geq  \frac{\epsilon}{4} \right) \\
& \stackrel{m \geq m_{\frac{\epsilon}{4},\delta}}{\leq }\delta.
\end{align*}
\end{proof}

\begin{proof}[\textbf{Proof of Theorem \ref{thm:best linear response}}]
The proof of this theorem is identical to the proof of \citeauthor{porat2017best} \cite[Theorem 2]{porat2017best}, and hence omitted.
\end{proof}

\begin{proof}[\textbf{Proof of Claim \ref{claim:flowerbound}}]

Denote $\mS=\left(x_1,\dots,x_m\right)\in \mX^m$ and $\bl r \in \R^m$ such that $\mH_i$ pseudo-shatters $\mS$ with witness $\bl r$. We prove the claim by showing that we can construct $\mS' \in \mZ^m$ that is shattered by $\mF_i$. For every binary vector $\bl b\in\{-1,1\}^m$ there exists $h_{\bl b}$ such that $\sign(h_b(x_j)-r_j)=b_j$ for every $j\in [m]$. 

Denote $H = \{h_{\bl b} \in \mH_i \}_{\bl b \in \{-1,1\}^m}$ such that $\abs{H} =2^m$. For every $j$ such that $r_j \geq 0$, let
\[
y_j = \min \left\{0,\min_{h_{\bl b} \in H}  h_{\bl b}(x_j)    \right\}.
\]
In addition, for every $j$ such that $r_j < 0$, let
\[
y_j = \max \left\{0,\max_{h_{\bl b} \in H}  h_{\bl b}(x_j)    \right\},
\]
and denote
\[
\mS' = \left( x_j,y_j,\abs{r_j} -\sign(r_j) t_j  \right)_{j=1}^m.
\]
We now show that $\mF_i$ shatters $\mS'$. Fix an arbitrary $\bl b\in\{-1,1\}^m$, and observe that in case $r_j \geq 0$
\begin{equation}
\label{eq:asdasdasd}
\begin{cases}
b_j = 1 \\
b_j = -1
\end{cases}
\Rightarrow
\begin{cases}
h_{\bl b}(x_j) \geq r_j  \\
h_{\bl b}(x_j) < r_j
\end{cases}
\Rightarrow
\begin{cases}
h_{\bl b}(x_j) -y_j \geq r_j -y_j \\
h_{\bl b}(x_j) -y_j < r_j -y_j 
\end{cases}
\Rightarrow
\begin{cases}
\abs{h_{\bl b}(x_j) -y_j }\geq r_j -y_j \\
\abs{h_{\bl b}(x_j) -y_j }< r_j -y_j
\end{cases},
\end{equation}
where the last argument holds since $h_{\bl b}(x_j)\geq y_j$. Alternatively, if $r_j <0$ we have
\begin{equation}
\label{eq:asdasdyju}
{\small
\begin{cases}
b_j = 1 \\
b_j = -1
\end{cases}
\Rightarrow
\begin{cases}
h_{\bl b}(x_j) \geq r_j  \\
h_{\bl b}(x_j) < r_j
\end{cases}
\Rightarrow
\begin{cases}
-h_{\bl b}(x_j)+y_j \leq -r_j +y_j \\
-h_{\bl b}(x_j) +y_j> -r_j+y_j
\end{cases}
\Rightarrow
\begin{cases}
\abs{-h_{\bl b}(x_j)+y_j} \leq \abs{r_j} +y_j \\
\abs{-h_{\bl b}(x_j)+y_j} > \abs{r_j} +y_j 
\end{cases}
}
,
\end{equation}
where again the last set of inequalities holds since $h_{\bl b}(x_j)\leq y_j$. 
In case one of Equations (\ref{eq:asdasdasd}) and (\ref{eq:asdasdyju}) holds in equality, we can slightly shift $r_j$ (as was done in the proof of Claim \ref{claim:auxiliaryg}); hence we assume that these are strict inequalities. The expression in Equation (\ref{eq:asdasdasd}) corresponds to $\mI(h_{\bl b},(x_j,y_j,r_j-y_j))$, while that of Equation (\ref{eq:asdasdyju}) corresponds to $\mI(h_{\bl b},(x_j,y_j,\abs{r_j}+y_j))$.

This analysis applies for every $\bl b$; hence, $\mF_i$ shatters $\mS'$ as required.
\end{proof}

\begin{proof}[\textbf{Proof of Claim \ref{claim:no-eq-direct}}]
\omer{proofread and verify}
We prove the claim by showing that for every $h \in \mH_1$ and $\epsilon>0$ there exists $h' \in \mH$ such that $\tpi_2(h,h')\geq\frac{2}{3}-\epsilon$. The assertion will then follow due to symmetry ($\mH_1=\mH_2$).

Let $h=(a,b)\in \mH_1$ denote an arbitrary strategy and fix $\epsilon>0$. We denote by $L,C,R$ the left, center and right squares of $U$ with respect to the horizontal axis, i.e.,
\[
L = \left( [0,1]\times [1,2]\right),\quad C = \left( [1,2]\times [0,1] \right), \quad R=\left( [2,3] \times [1,2] \right).
\] 
In addition, let $\ind_L(h)$ indicate whether $h$ crosses $L$, namely
\[
\ind_L(h) = 
\begin{cases}
1 & \{(x,y) \mid ax+b=y  \} \cap L \neq \emptyset \\
0 & \text{otherwise}
\end{cases},
\]
and similarly for $\ind_C(h), \ind_R(h)$. We proceed by an exhaustive case analysis over the alternatives for $\ind_L(h)+\ind_C(h)+\ind_R(h)$:
\begin{itemize}
\item If $\ind_L(h)+\ind_C(h)+\ind_R(h)\leq 1$, by taking $h'$ to be parallel to $h$ and close enough, player 2 can get at least $\frac{2}{3}$.
\item Else, if $\ind_L(h)+\ind_C(h)+\ind_R(h) = 2$, we have three sub-cases:
\begin{itemize}
\item In case $\ind_L(h)=\ind_R(h)=1$, denote by $L_{ab}$ ($L_{bl}$) the set of points above (below) $h$ inside $L$, namely
\[
L_{ab} = \left\{ (x,y)\in L \mid ax+b  \leq y \right\}, L_{bl} = \left\{ (x,y)\in L \mid ax+b  > y \right\}.
\]
Similarly,  let $R_{ab}$ ($R_{bl}$) denote the set of points above (below) $h$ inside $R$. Let $\lambda(A)$ denote the Lebesgue measure (informally, the area) of a set $A\subset \R^2$. Notice that by selecting the appropriate $h'$, player 2 can obtain a payoff arbitrarily close to 
\[
\frac{1}{3} \left(\lambda(C)+\max\left\{\lambda(L_{bl})+\lambda(R_{ab}),\lambda(L_{ab})+\lambda(R_{bl}) \right\} \right) \geq \frac{2}{3}   ,
\]
where the inequality follows from having $\lambda(L_{bl})+\lambda(R_{ab})+\lambda(L_{ab})+\lambda(R_{bl})=\lambda(L)+\lambda(R)=2$.
\item In case $\ind_L(h)=\ind_C(h)=1$, let $L_{ab},L_{bl}$ be as in the previous case, and also define $C_{ab},C_{bl}$ in the same manner. Here too, player 2 can obtain a payoff arbitrarily close to 
\[
\frac{1}{3} \left(\lambda(R)+\max\left\{\lambda(L_{bl})+\lambda(C_{bl}),\lambda(L_{ab})+\lambda(C_{ab}) \right\} \right) \geq \frac{2}{3}   ,
\]
where the inequality follows from having $\lambda(L_{bl})+\lambda(C_{bl})+\lambda(L_{ab})+\lambda(C_{ab})=\lambda(L)+\lambda(C)=2$.
\item The case of $\ind_R(h)=\ind_C(h)=1$ is symmetric to the previous case.
\end{itemize}
\item Finally, if $\ind_L(h)+\ind_C(h)+\ind_R(h)=3$, then $h=(0,1)$. In this case, by taking $h'=(0,1+3\epsilon)$, player 2 obtains precisely $\frac{2}{3}-\epsilon$.
\end{itemize}
\end{proof}

\begin{proof}[\textbf{Proof of Lemma \ref{lemma:realizable}}] \omer{proofread and verify}
Let $\tilde G = \langle \tmZ,\tmD ,\mN,\mH,\tpi \rangle$. We define an auxiliary game $G = \langle \mZ,\mD ,\mN,\mH,\pi \rangle$ such that $\mT=\{0\}$, $\mZ=\tmZ \times \mT$ and $\mD$ such that
\[
\forall Z\subseteq \tmZ: \mD(Z\times \mT)=\tmD(Z).
\]
We use the tilde notation to describe the elements of $\tilde G$, and by non-tilde to denote the elements of $G$. 

Recall that $\tilde G$ is realizable, and let $h_i^*\in \mH_i$ such that $h_i^*(x)=y$ for every $(x,y)\in \tmZ$. By definition of $\bl h$, which is a function of the sample $\tmS$,  $\pi_i^{\mS}(h_i^*,\bl h_{-i}) \leq\pi_i^{\mS}(\bl h)$ holds in $G$. By slightly modifying the proof of Lemma \ref{lemma:empiseq} (the first inequality, as $\bl h$ is a PNE of the empirical game of $G$ and not an $\frac{\epsilon}{2}$-PNE as analyzed there), we know that with probability of at least $1-\delta$ over the selection of $\mS$ it holds that
\begin{equation}
\label{eq:hghhgsd}
\forall i\in \mN: \sup_{h_i'\in \mH_i} \pi_i(h_i', \bl h_{-i})-\pi_i(\bl h) = \pi_i(h_i^*, \bl h_{-i})-\pi_i(\bl h)  < \epsilon.
\end{equation}
Using the definitions of $\tw_i,w_i$ and Equation (\ref{eq:hghhgsd}) we obtain that
\begin{align}
\label{eq:poire}
\tpi_i(h_i^*, \bl h_{-i})-\tpi_i(\bl h) &= \E_{(x,y)\sim \tmD} \left[\tw_i(x,y;h_i^*, \bl h_{-i})-\tw_i(x,y;\bl h)  \right] \nonumber\\
&= \E_{(x,y)\sim \tmD} \left[\frac{\ind_{h_i(x)\neq y} }{\abs{B(\tz ; h_i^*, \bl h_{-i})}}   \right]\nonumber\\
&= \E_{(x,y,0)\sim \mD} \left[\frac{\ind_{h_i(x)\neq y} }{\abs{B(\tz ; h_i^*, \bl h_{-i})}}   \right]\nonumber\\
&= \E_{(x,y,0)\sim \mD} \left[\frac{\ind_{h_i(x)\neq y} }{1+\sum_{i'\in\mN\setminus \{i\}} \mI(x,y,0,h_{i})}   \right]\nonumber\\
&= \E_{(x,y,0)\sim \mD} \left[\ind_{h_i(x)\neq y}w_i(x,y,0;h_i^*, \bl h_{-i}) \right]\nonumber\\
&= \E_{(x,y,0)\sim \mD} \left[w_i(x,y,0;h_i^*, \bl h_{-i})-w_i(x,y,0; \bl h) \right]\nonumber\\
&= \pi_i(h_i^*, \bl h_{-i})-\pi_i(\bl h) \nonumber\\
&<\epsilon,
\end{align}
where the last inequality holds with probability at least $1-\delta$. As a result, with probability at least $1-\delta$,
\begin{equation}
\label{eq:bbgbgu}
\forall i\in \mN: \sup_{h_i'\in \mH_i} \tpi_i(h_i', \bl h_{-i})-\tpi_i(\bl h) = \tpi_i(h_i^*, \bl h_{-i})-\tpi_i(\bl h)  < \epsilon.
\end{equation}
This concludes the proof of the lemma.
\end{proof}

\section{Additional claims and proofs}
\label{sec:additionalproofs}
\begin{claim}
\label{claim:rqep}
For a given $\bl h$ and  $m \geq \frac{2}{\epsilon^2 }$ it holds that
\begin{align}
\label{eq:hoefbdvg}
\Pr_{\mS' \sim \mD^m}\left(\abs{\pi_i(\bl h)-\pi_i^{\mS'}(\bl h)}  \leq \frac{\epsilon}{2}   \right) \geq \frac{1}{2}.
\end{align}
\end{claim}
\begin{proof}[\textbf{Proof of Claim \ref{claim:rqep}}]
Recall Chebyshev's inequality
\[
\Pr \left(\abs{X - \E[x]} \geq \epsilon \right) \leq \frac{\Var(X)}{\epsilon^2}.
\]
Applying it for our problem, we get
\[
\Pr_{\mS' \sim \mD^m}\left(\abs{\pi_i(\bl h)-\pi_i^{\mS'}(\bl h)}  \geq \frac{\epsilon}{2}   \right) \leq \frac{\Var(\pi_i^{\mS'}(\bl h))}{\frac{\epsilon^2}{4}}.
\]
Notice that $\pi_i^{\mS'}(\bl h)$ is the average of independent random variables bounded in the $[0,1]$ segment; hence, by Popoviciu's inequality on variances we have 
\[
\Var(\pi_i^{\mS'}(\bl h)) \leq \frac{1}{4m}.
\]
Finally, for $m \geq \frac{2}{\epsilon^2 }$ it holds that
\[
\Pr_{\mS' \sim \mD^m}\left(\abs{\pi_i(\bl h)-\pi_i^{\mS'}(\bl h)}  \leq \frac{\epsilon}{2}   \right)=1-\Pr_{\mS' \sim \mD^m}\left(\abs{\pi_i(\bl h)-\pi_i^{\mS'}(\bl h)}  \geq \frac{\epsilon}{2}   \right) \geq 1-\frac{\frac{1}{4m}}{\frac{\epsilon^2}{4}}=1-\frac{1}{\epsilon^2m}\geq \frac{1}{2}.
\]
\end{proof}

\begin{claim} 
\label{claim:auxhoeffbin}
Let $m\geq 15$, $(X_i)_{i=1}^m$ be a sequence of i.i.d. Bernoulli r.v. with $p=\frac{1}{2}$, and let $\bar X = \frac{1}{m} \sum_{i=1}^m X_i$. Then $\Pr\left(\frac{1}{2} < \bar X < \frac{3}{4}\right) \geq \frac{1}{4}$ .
\end{claim}
\begin{proof}[\textbf{Proof of Claim \ref{claim:auxhoeffbin}}]
By Hoeffding's inequality we have $\Pr\left( \bar X \geq (p+\epsilon) \right)\leq \exp\left( -2\epsilon^2 m \right) $. Therefore
\begin{align*}
\Pr\left(\frac{1}{2} < \bar X < \frac{3}{4}\right) = \Pr\left( \bar X < \frac{3}{4}\right) - \Pr\left(\bar X \leq \frac{1}{2} \right) = 1 - \Pr\left( \bar X \geq \frac{3}{4}\right) -\frac{1}{2} \stackrel{\epsilon= \frac{1}{4}}{\geq} \frac{1}{2}-e^{\frac{-m}{8} } \stackrel{ m\geq 15}{\geq} \frac{1}{4}.
\end{align*}
\end{proof}

\section{Details for Subsection \ref{SUBSEC:SIM}}
\label{sec:simulation specifics}
In all simulations, we selected $m=100$ points. The distribution over $\mX \times \mY \times \mT$ was the product distribution $\mD_{\mX} \cdot \mD_{\mY} \cdot \mD_{\mT}$, such that $\mD_{\mX}$ is $Uni[0,5]$ and $\mD_{\mY}$ as appear in Subsection \ref{SUBSEC:SIM} with the following parameters:
\begin{enumerate}
\item Linear: For every instance $x$, $y=2x+1+\epsilon$ for $\epsilon\sim Normal(0,1)$.
\item V-shape: The value of every $x$ is determined by $y=5\abs{x-2.5}+1+\epsilon$ for $\epsilon\sim Normal(0,1)$.
\item X-shape: For every instance $x$, the value $y$ is distributed as follows:
\[
y=
\begin{cases}
5x+0+\epsilon & \text{w.p. }\frac{1}{2}\\
-5x+25+\epsilon & \text{w.p. }\frac{1}{2}
\end{cases},
\]
where $\epsilon\sim Normal(0,1)$.
\item Piecewise: For every $x$, the value $y$ is given by
\[
y=
\begin{cases}
10x+1+\epsilon & \text{if }x\leq 2.5\\
2x+18+\epsilon & \text{else}
\end{cases},
\]
where $\epsilon\sim Normal(0,1)$.
\end{enumerate}
In addition, the tolerance level was deterministic, i.e., $\mD_{\mT}$ was the degenerate deterministic distribution over a value $t$, such that 
\begin{enumerate}
\item Linear: High - $t=\frac{3}{2}$ ,Medium - $t=\frac{2}{3}$, Low - $t=\frac{1}{3}$.
\item V-shape: High - $t=\frac{2}{1}$ ,Medium - $t=\frac{1}{1}$, Low - $t=\frac{1}{2}$.
\item X-shape: High - $t=\frac{2}{1}$ ,Medium - $t=\frac{1}{1}$, Low - $t=\frac{1}{2}$.
\item Piecewise: High - $t=\frac{3}{1}$ ,Medium - $t=\frac{3}{2}$, Low - $t=\frac{1}{2}$.
\end{enumerate}
}\fi} %closing appendices

\end{document}

%% file: defs.tex
\usepackage{xcolor}
\usepackage{gensymb} %for degree - in the acknowledge men
\usepackage{graphicx}
\usepackage{amsmath}
\usepackage{dsfont}
\usepackage{multirow} %merge rows in the table
\usepackage{amssymb} %triangleq

\newcount\Comments  % 0 suppresses notes to selves in text
\newcount\Includeappendix %include appendices or not

\definecolor{darkgreen}{rgb}{0,0.6,0}
\newcommand{\kibitz}[2]{\ifnum\Comments=1{\color{#1}{#2}}\fi}

\newcommand{\omer}[1]{\kibitz{blue}{[Omer :#1]}}

\makeatletter
\newcommand{\RemoveAlgoNumber}{\renewcommand{\fnum@algocf}{\AlCapSty{\AlCapFnt\algorithmcfname}}}
\newcommand{\RevertAlgoNumber}{\algocf@resetfnum}
\makeatother
\newcommand{\mX}{\mathcal{X}}

\newcommand{\mI}{\mathcal{I}}
\newcommand{\mF}{\mathcal{F}}
\newcommand{\mY}{\mathcal{Y}}
\newcommand{\mZ}{\mathcal{Z}}
\newcommand{\mH}{\mathcal{H}}

\newcommand{\mG}{\mathcal{G}}
\newcommand{\mD}{\mathcal{D}}
\newcommand{\mR}{\mathcal{R}}

\newcommand{\mN}{\mathcal{N}}
\newcommand{\mS}{\mathcal{S}}
\newcommand{\mT}{\mathcal{T}}
\newcommand{\mW}{\mathcal{W}}
\newcommand{\ind}{\mathds{1}}

\newcommand{\R}{\mathbb{R}}

\newcommand{\defeq}{\stackrel{\text{def}}{=}}
\newcommand{\vc}{\textnormal{VCdim}}
\newcommand{\pdim}{\textnormal{Pdim}}

\newcommand{\Var}{\mathrm{Var}}

\newcommand{\tmS}{\tilde{\mathcal{S}}}

\newcommand{\tmZ}{\tilde{\mathcal{Z}}}
\newcommand{\tmD}{\tilde{\mathcal{D}}}

\newcommand{\tz}{\tilde{z}}
\newcommand{\tw}{\tilde{w}}
\newcommand{\tpi}{\tilde{\pi}}

\newcommand{\inprod}[2]{\langle\, #1,#2\rangle}

\newtheorem{theorem}{Theorem}[section]
\newtheorem{lemma}[theorem]{Lemma}

\newtheorem{proposition}[theorem]{Proposition}
\newtheorem{claim}[theorem]{Claim}

\newtheorem{corollary}[theorem]{Corollary}

\newtheorem{definition}{Definition}[section]
\theoremstyle{remark}
\newtheorem{example}{Example}[section]

\newcommand\bl[1]{\boldsymbol{ #1 } }
\newcommand{\norm}[1]{\left\lVert#1\right\rVert}
\newcommand\abs[1]{\left| #1  \right|}
\DeclareMathOperator*{\argmax}{arg\,max} %argmax

\DeclareMathOperator{\E}{\mathbb{E}}

\DeclareMathOperator{\sign}{sign}